\documentclass[12pt]{article}

\usepackage{url}
\usepackage{xcolor}
\usepackage{color}
\usepackage{graphicx}
\usepackage{amssymb}
\usepackage{amsmath,amsthm,paralist}
\usepackage{amsbsy}
\usepackage{cite}
\usepackage{algorithm,algorithmic}
\usepackage[margin=1in]{geometry}
\usepackage{afterpage}
\usepackage{multirow}
\usepackage{caption}
\usepackage{bbm}
\usepackage{mathtools}
\usepackage{subcaption}
\usepackage{soul}
\usepackage{wrapfig}
\usepackage{epstopdf}

\usepackage{hyperref}
\hypersetup{
    colorlinks=true, %set true if you want colored links
    linktoc=all,     %set to all if you want both sections and subsections linked
    linkcolor=blue,  %choose some color if you want links to stand out
}

\allowdisplaybreaks

\newtheorem{lemma}{Lemma}

\newtheorem{theorem}{Theorem}
\newtheorem{result}{Result}

\theoremstyle{remark}
\newtheorem{remark}{Remark}

\newcommand{\eps}{\epsilon}
\newcommand{\bitem}{\begin{itemize}}
\newcommand{\eitem}{\end{itemize}}

\newcommand{\beqn}{\begin{equation}}
\newcommand{\eeqn}{\end{equation}}
\newcommand{\balign}{\begin{align}}
\newcommand{\ealign}{\end{align}}

\newcommand{\inner}[1]{\left<#1\right>}

\makeatletter
\renewcommand\paragraph{\@startsection{paragraph}{4}{\z@}%
            {-2.5ex\@plus -1ex \@minus -.25ex}%
            {1.25ex \@plus .25ex}%
            {\normalfont\normalsize\bfseries}}
\makeatother
\usepackage{authblk}
\begin{document}
\title{An Approximate Message Passing Framework for Side Information\thanks{
AM and DN were supported by NSF CAREER $\#1348721$, and NSF BIGDATA $\#1740325$;
DB and YZ were supported by NSF EECS $\#1611112$.
A subset of our results appeared in Baron et al.~\cite{BNMRW2017}.
}}
\renewcommand\footnotemark{}
\author[1]{Anna~Ma}
\author[2]{You~(Joe)~Zhou}
\author[3]{Cynthia~Rush}
\author[2]{Dror~Baron}
\author[4]{Deanna~Needell}

\affil[1]{Department of Mathematics, University of California San Diego}
\affil[2]{Department of Electrical and Computer Engineering, NC State University}
\affil[3]{Department of Statistics, Columbia University}
\affil[4]{Department of Mathematics, University of California Los Angeles}

\maketitle

\begin{abstract}
Approximate message passing (AMP) methods have gained recent traction in sparse signal recovery. Additional information about the signal, or \emph{side information} (SI), is commonly available and can aid in efficient signal recovery. This work presents an AMP-based framework that exploits SI and can be readily implemented in various settings for which the SI results in separable distributions. To illustrate the simplicity and applicability of our approach, this framework is applied to a Bernoulli-Gaussian (BG) model and a time-varying birth-death-drift (BDD) signal model, motivated by applications in channel estimation.
We develop a suite of algorithms, called AMP-SI, and derive denoisers for the
BDD and BG models. Numerical evidence demonstrating the advantages of our approach are presented alongside empirical evidence of the accuracy of a proposed state evolution.
\end{abstract}

\section{Introduction}

The core focus of research in many disciplines, including but not limited to
communication~\cite{Caire2004},
compressive imaging~\cite{Arguello2011},
matrix completion~\cite{Candes2009},
quantizer design~\cite{Lloyd82},
large-scale signal recovery~\cite{ZhuBaronMPAMP2016ArXiv}, and
sparse signal processing~\cite{CandesUES}, is on accurately recovering a high-dimensional, unknown signal from a limited number of noisy linear measurements by exploiting
probabilistic characteristics and structure in the signal.

We consider the following model for this task. For an unknown signal $x \in \mathbb{R}^N$,
\begin{equation}
y= Ax + z,
\label{eq:hdreg}
\end{equation}
where $y \in \mathbb{R}^M$ are
noisy measurements,
$A \in \mathbb{R}^{M \times N}$ is the measurement
matrix, and $z \in \mathbb{R}^M$ is measurement noise.
The objective of signal recovery is to recover
or estimate $x$ from knowledge of only $y$ and $A$, and in some cases statistical
knowledge about $x$ and $z$.
A great deal of effort has gone into developing schemes for such signal recovery, for example $\ell_1$ minimization based approaches for sparse recovery~\cite{DonohoBP,LASSO1996} and computationally efficient iterative algorithms~\cite{DonohMM_Message,RanganADMMGAMP2015_ISIT, Boyd2011},
and supporting theory to tackle these challenges as
datasets become larger and multidimensional.
For scenarios in which the signal's prior distribution is available, the approximate message passing framework is often utilized.

\subsection{AMP for Signal Recovery}
Approximate message passing or AMP~\cite{DonohMM_Message, Krzakala2012probabilistic, RanganGAMP2010}
is a low-complexity algorithmic framework for efficiently recovering sparse signals in high-dimensional regression tasks (\ref{eq:hdreg}).
AMP algorithms are derived as Gaussian or quadratic approximations of loopy belief propagation algorithms (e.g., min-sum, sum-product)
on the dense factor graph corresponding to~\eqref{eq:hdreg}.

AMP has a few features that make it attractive for signal recovery.
In certain problem settings, AMP offers convergence in linear time, and
its performance can be tracked accurately with a simple scalar iteration known as state evolution (SE), discussed below.
In addition, it is well-accepted that the performance of AMP will be no worse than the best polynomial-time algorithms available \cite{maleki2010approximate}.

\textbf{AMP algorithm:}
The standard AMP algorithm \cite{DonohMM_Message} iteratively updates estimates of the unknown
input signal, with $x^t \in \mathbb{R}^N$ being the estimate at iteration $t$.  The algorithm is given by the following set of updates. Assume that $x^0$ is the all-zero vector and update for $t\geq 0$ with the following iterations:
\begin{align}
&r^t = y - Ax^t + \frac{r^{t-1}}{\delta} \inner{\eta_{t-1}^\prime( x^{t-1} + A^Tr^{t-1})}, \label{eq:AMP1}\\
&x^{t+1} = \eta_t( x^t + A^Tr^t). \label{eq:AMP2}
\end{align}
Note that $\eta_t\colon \mathbb{R} \rightarrow \mathbb{R}$ is an appropriately-chosen sequence of functions and $\delta = \frac{M}{N}$ is the measurement rate.
The functions $\{\eta_t(\cdot)\}_{t \geq 0}$ act element-wise on
their vector inputs and have derivatives
$\eta_t^\prime(w) = \frac{\partial}{\partial w}\eta_t(w)$.  Moreover, $\inner{w} = \frac{1}{N} \sum_{i=1}^N w_i$
is the empirical mean, where $w \in \mathbb{R}^N$.
Here and throughout, we use capital letters to represent {\em random variables} (RVs) and lower case letters to represent realizations.
We also denote a Gaussian RV
with mean $\mu$ and variance $\sigma^2$ by $\mathcal{N}(\mu, \sigma^2)$.

Assume that the measurement matrix $A$ has independent and identically distributed (i.i.d.) $\mathcal{N}(0, 1/M)$ entries and the entries of the signal $x$ are i.i.d.\ $\sim f(X)$, where $f(X)$ is the probability density function (pdf) of the signal.  Under these assumptions, one useful feature of AMP is that the input to the denoiser, $x^t + A^Tr^t$,
which we refer to as the \emph{pseudo-data}, is almost surely equal in distribution, in the large system limit as $N \rightarrow \infty$
with fixed $\delta$, to the true signal $x$ plus i.i.d.\ Gaussian noise with variance $\lambda_t^2$, where $\lambda_t^2$ is
a constant value given by the SE equations, introduced in \eqref{eq:AMP_SE} below.
These favorable statistical properties of the pseudo-data are due to the presence of the `Onsager' term,
$ \frac{r^{t-1}}{\delta} \inner{\eta_{t-1}^\prime(x^{t-1} + A^Tr^{t-1})}$,
used in the residual step \eqref{eq:AMP1} of the AMP updates.

\textbf{State evolution (SE):} One of AMP's attractive features is that under suitable conditions on $A$ and $x$, its performance can be tracked accurately with a simple scalar
iteration referred to as state evolution (SE)~\cite{BayMont11,RushV16}. In particular, performance measures such as the
$\ell_2$-error or the $\ell_1$-error in the algorithm's iterations concentrate to constants predicted by
SE.  Let the noise $z$ of \eqref{eq:hdreg} be element-wise i.i.d. $\sim f(z)$ and for $Z \sim f(Z)$ let $\sigma_z^2 = \mathbb{E}[Z^2]$.  Then the SE equations follow:
let $\lambda_0 = \sigma_z^2 + \mathbb{E}[X^2]/\delta$ and for $t \geq 0$,
\begin{equation}
\lambda_t^2 = \sigma_z^2 + \frac{1}{\delta}\mathbb{E}\left[\left(\eta_{t-1}(X +
\lambda_{t-1}U) - X\right)^2\right],
\label{eq:AMP_SE}
\end{equation}
where $X \sim f(X)$ is independent of  $U \sim \mathcal{N}(0,1)$ and $\lambda_t^2$ tracks the variance of the difference between the pseudo-data and signal at iteration $t$.

The AMP updates \eqref{eq:AMP1} - \eqref{eq:AMP2} rely on appropriately-chosen {\em denoisers} $\{ \eta_t\}_{t \geq 0}$, which reduce the noise in the optimization task at each iteration. 
Owing to the favorable properties of the psuedo-data and the fact that one is often interested in evaluating the performance of the algorithm using the mean squared error (MSE), $\eta_t$ in iteration $t$ is often chosen to be
the minimum mean squared error (MMSE) denoiser
based on the pdf of $x$:
\begin{equation}
\eta_t(a)=\mathbb{E}[X | X + \lambda_t U = a],
\label{eq:cond_dist_denoiser}
\end{equation}
where $U \sim \mathcal{N}(0,1)$, and $X \sim f(X)$ is a RV
with the same pdf as that of $x$. See Section \ref{sec:SEexplain} for further insights of how SE behaves in our framework.

\subsection{Side information} \label{sec:si-eg}
In information theory~\cite{Cover06,Gallager68},
it is well known that when different communication systems share {\em side information} (SI), overall communication can happen more efficiently.
As an example, when running a Bayesian signal recovery algorithm on an input $x$ with an unknown probability density,
feedback about the estimated density
leads to improved signal recovery quality~\cite{Kamilov2012}.

Signal recovery algorithms often have access to SI,
denoted $\widetilde{x}$, that,
as we will soon see, offers the potential to
markedly improve recovery quality. For the noisy linear model of~\eqref{eq:hdreg}, SI has been shown to aid signal recovery when considering various application settings~\cite{wang2015approximate,SaabM_weighted,NeedellWeighted16,VanLuong2017, Renna2016, Herzat2013, Vaswani2010, Weizman2015, Chen2008, Mota2016, Mota2017}.
For example, three dimensional (3D)
video acquisition could be performed by acquiring each frame of video,
which is a 2D image, independently of other frames
using a single pixel camera~\cite{takhar2006new}.
While recovering the current frame, it is likely that
one is simultaneously recovering
the previous and next frames, which
can be used as SI.

We will demonstrate that our approach is potentially useful in applications
by studying a channel estimation problem in wireless communication systems
(Fig.~\ref{fig:batch}).
In typical channel estimation scenarios, a wireless device transmits a pilot sequence and data payload in {\em batches}.
In batch~$b$, the pilot sequence $p$ is transmitted into the channel, where it is convolved
with the channel response $x^b$, yielding noisy linear measurements
(details in Section~\ref{subsec:channel_est}).
Not only is the channel response $x^b$ in batch~$b$ sparse,
the slowly time varying nature of the channel ensures that its
differences relative to channel responses in previous batches are
structured. Therefore, we can use $\widetilde{x} = \widehat{x}^{b-1}$, the channel
response estimated in the previous batch, as SI while estimating $x^b$ in
the current batch. In Section~\ref{sec:sims}, we demonstrate that SI in the above-mentioned batched manner
helps AMP achieve lower MSE for a model motivated by channel estimation.

\begin{figure}[t!]
\centering
\includegraphics[width=.4\textwidth]{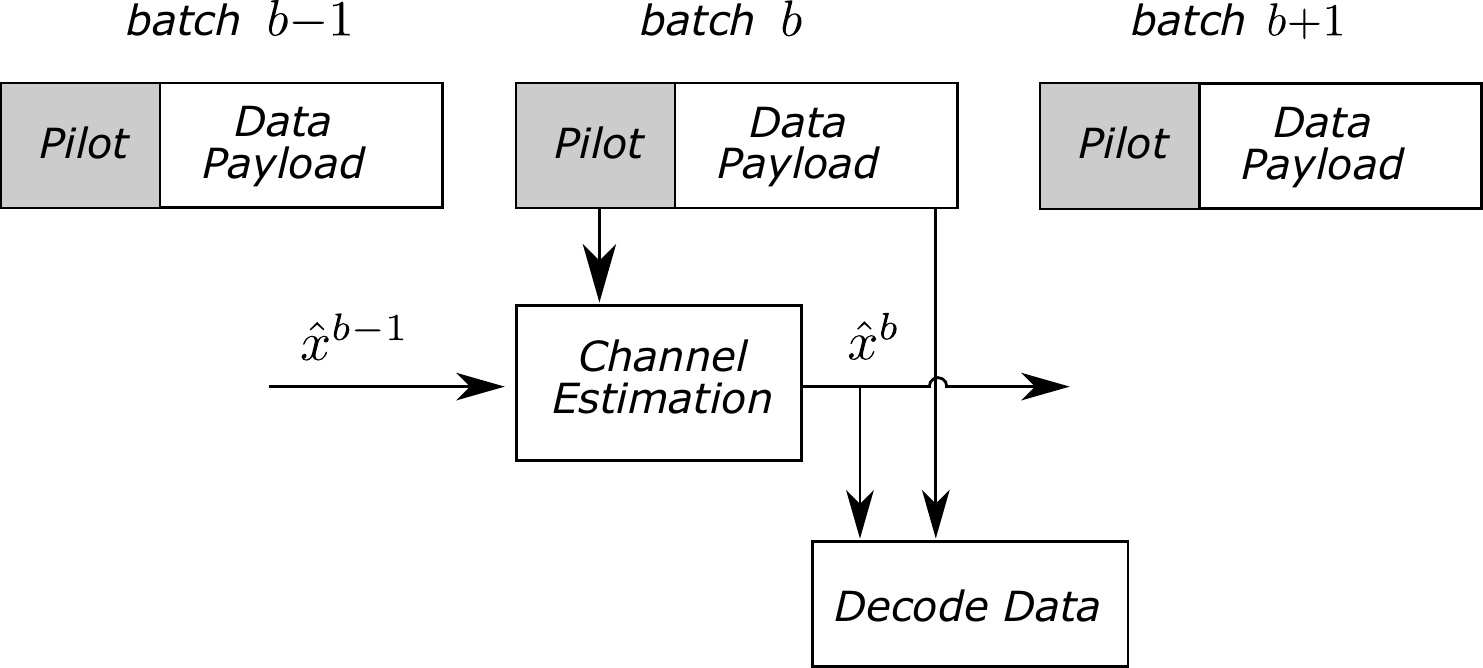}
\caption{In batch~$b$, the wireless device transmits a pilot and data payload.
The channel filter $x^b$ is estimated using the channel's response to
the pilot along with SI $\widetilde{x} = \widehat{x}^{b-1}$, the channel filter estimated in the previous batch. The estimated $\widehat{x}^b$ is used to decode the data and as SI in the next batch to estimate $x^{b+1}$.
}\label{fig:batch}
\vspace*{-4.5mm}
\end{figure}

\subsection{Contributions and Organization}
In this work we develop a class of
sparse signal recovery
algorithms that integrate SI into AMP. 
Our main contribution is a
framework that incorporates SI in the denoiser of AMP that is Bayes-optimal in certain cases and
can be adapted to
arbitrary dependencies between the signal and the SI.
Moreover, our framework's conceptual simplicity allows
us to extend existing SE results to AMP-SI as in \eqref{eq:AMP_SE};
these SE results for signal recovery with
SI are lacking in prior work~\cite{wang2015approximate},~\cite{DCSAMP}.
In the case where the SI is a Gaussian-noise corrupted view of the true signal, we rigorously show Bayes-optimality properties for AMP-SI.
For more general cases, we demonstrate empirically that our proposed SE formulation
tracks the AMP performance.

We demonstrate our framework through its application to two types of signals. First, a Bernoulli-Gaussian (BG) signal and second, motivated by the channel estimation problem discussed in Section \ref{sec:si-eg}, a time-varying birth-death-drift (BDD) signal. Our numerical experiments show that
our proposed framework achieves a lower MSE than other previously-studied SI methods.

The remainder of the paper is organized as follows.
In Section \ref{sec:amp-si}, we discuss the AMP algorithm and prior work in AMP approaches that utilize SI. We then present our AMP framework for SI. Next we discuss the BG model in Section \ref{sec:BG}, which is a simplified version of the BDD model studied in Section~\ref{sec:BDD}. In Section~\ref{sec:sims}, we include numerical simulations demonstrating the good performance of AMP-SI. Section~\ref{sec:futurework} concludes.

\section{AMP with Side Information}
\label{sec:amp-si}

\subsection{Prior Work}
While integrating SI (or prior information) into signal recovery algorithms is not new~\cite{wang2015approximate, VanLuong2017, Renna2016, Mota2017, Chen2016,manoel2017streaming}, 
our work is a unified framework within AMP
that supports arbitrary dependencies between the $(X_n,\widetilde{X}_n)_{n=1}^N$ pairs.
Prior work using SI has been either heuristic, limited 
to specific applications, or outside the AMP framework. For example, Wang and Liang~\cite{wang2015approximate} proposed the Generalized Elastic Net Prior approach, which integrates SI into AMP for a specific signal prior density, but the method lacks Bayes optimality properties and is difficult to apply to other signal models.  Our algorithmic framework overcomes these limitations through a generalized, Bayes optimal framework.

Ziniel and Schniter~\cite{DCSAMP} developed an AMP-based signal recovery algorithm, namely DCS-AMP, for a
time-varying signal model based on Markov processes for the support and amplitude.
The Markov processes and corresponding dependencies between variables are captured by factor graph models.
While our BDD model (details in Section~\ref{sec:BDD}) is closely related to their
time-varying signal model, our emphasis is to introduce the AMP-SI framework
and demonstrate how SI can be incorporated in AMP without needing to carefully craft
factor graphs for every new signal model.

Manoel et al. implemented an AMP-based algorithm called MINI-AMP in which
the input signal is repeatedly reconstructed in a streaming fashion,
and information from past reconstruction attempts is aggregated into a prior,
thus improving ongoing reconstruction results~\cite{manoel2017streaming}.
Interestingly, the signal recovery approach of MINI-AMP resembles that of AMP-SI,
in particular when the BG model of Section \ref{sec:BG} is used.
Finally, Manoel et al. proved that MINI-AMP is MMSE-optimal~\cite{manoel2017streaming}.

\subsection{Our Approach: AMP-SI}
In this paper we introduce an algorithmic framework that utilizes available SI.
Our SI takes the form of an estimate $\widetilde{x} \in \mathbb{R}^N$, which is statistically dependent on the signal $x$ through some joint pdf $f(X, \widetilde{X})$.
We propose a \emph{conditional denoiser},
\begin{equation}
\label{eq:eta_CAMP}
\eta_t(a, b)=\mathbb{E}[X | X + \lambda_t U = a,
\widetilde{X}=b],
\end{equation}
where $U\sim\mathcal{N}(0,1)$ is independent of $(X, \widetilde{X})\sim f(X, \widetilde{X})$. The denoiser provides an MMSE estimate of the signal while incorporating SI. We refer to our framework using the proposed denoiser \eqref{eq:eta_CAMP} within the standard AMP algorithm \eqref{eq:AMP1} - \eqref{eq:AMP2} as the AMP-SI method. The AMP-SI algorithm is the following.  Assume $x^0 $ is the all-zero vector and update for $t\geq 0$:
\begin{align}
&r^t = y - Ax^t + \frac{r^{t-1}}{\delta} \inner{\eta_{t-1}^\prime( x^{t-1} + A^Tr^{t-1}, \widetilde{x})},  \label{eq:CAMP1}\\
&x^{t+1} = \eta_t( x^t + A^Tr^t, \widetilde{x}).  \label{eq:CAMP2}
\end{align}
Note that $\eta_t(\cdot, \cdot)$ is the denoising function
proposed in \eqref{eq:eta_CAMP}, its derivative $\eta_t^\prime(w, \cdot) = \frac{\partial}{\partial w}\eta_t(w, \cdot)$ is with respect to the first input, and $\inner{w} = \frac{1}{N} \sum_{i=1}^N w_i$ for $w \in \mathbb{R}^N$. The $\lambda_t$ value in \eqref{eq:eta_CAMP} is given by SE equations for AMP-SI.  Again, let the noise $z$ be element-wise i.i.d. $\sim f(z)$ and for $Z \sim f(Z)$, let $\sigma_z^2 = \mathbb{E}[Z^2]$. Then, $\lambda_0 = \sigma_z^2 + \mathbb{E}[X^2]/\delta$ and for $t \geq 0$,
\begin{equation}
\lambda_t^2 = \sigma_z^2 + \frac{1}{\delta}\mathbb{E}\Big[\Big(\eta_{t-1}(X + \lambda_{t-1}U, \widetilde{X}) - X\Big)^2\Big],
\label{eq:SE2}
\end{equation}
where $(X, \widetilde{X}) \sim f(X, \widetilde{X})$ are independent of $U$, which is a standard Gaussian RV.
In comparison to standard AMP, the conditional denoiser function $\eta_t(\cdot, \cdot)$ uses SI to denoise the pseudo-data in AMP-SI.

We note that while there are rigorous theoretical results~\cite{BayMont11,RushV16} proving that for large $N$ the pseudo-data  is approximately equal (in distribution) to the true signal $x$ plus i.i.d.\ Gaussian noise with variance $\lambda_t^2$ in the case of standard AMP~\eqref{eq:AMP1} - \eqref{eq:AMP2} with the standard SE~\eqref{eq:AMP_SE}, we only \emph{conjecture} that such a result is true for AMP-SI~\eqref{eq:CAMP1} - \eqref{eq:CAMP2} with the corresponding SE~\eqref{eq:SE2}. This conjecture is supported by empirical evidence in Section~\ref{sec:sims} that shows that the SE accurately tracks the MSE of the AMP-SI estimates, and by a theoretical proof relating to the $\ell_2-$error (see Section ~\ref{sec:theory}).

To show that AMP-SI is conceptually intuitive to apply and can improve signal estimation quality in applications where SI is available,
we apply AMP-SI to a preliminary channel estimation model (Section~\ref{sec:BDD}). More complex models, like those that incorporate element-wise dependencies between signal and SI, not only require more complicated denoiser and SE derivations but also need to be handled carefully theoretically.
While using more realistic channel models is left for future work,
our encouraging numerical results show that AMP-SI can be used beyond toy models
such as BG (Section~\ref{sec:BG}).
\subsection{AMP-SI Theory} \label{sec:theory}

\subsubsection{State Evolution Analysis}\label{sec:SEexplain}

As mentioned previously, the performance of AMP~\eqref{eq:AMP1}-\eqref{eq:AMP2} at each step of the algorithm can be \emph{rigorously} characterized by the SE equations in~\eqref{eq:AMP_SE}.  When the empirical density function
of the unknown signal $x$ converges to some pdf $f(X)$ on $\mathbb{R}$ and the denoisers $\{\eta_t(\cdot)\}_{t \geq 0}$ used in the AMP updates are applied element-wise to their input, Bayati and Montanari~\cite{BayMont11} proved
that the SE accurately predicts AMP performance in the large system limit.  For example, their result implies that the MSE, $\frac{1}{N} ||x^t - x||^2$, equals $\delta(\lambda_t^2 - \sigma_z^2)$ almost surely in the large system limit, and additionally, it characterizes the limiting constant values for a fairly general class of loss functions. Rush and Venkataramanan~\cite{RushV16} provide a concentration version of the asymptotic result when the prior density of $x$ is i.i.d.\ sub-Gaussian, showing that the probability of $\epsilon$-deviation between various performance measures and their limiting constant values fall exponentially in $N$.

Considering AMP-SI, however, we cannot directly apply the theoretical results of Bayati and Montanari~\cite{BayMont11} or Rush and Venkataramanan~\cite{RushV16}.  Each entry $n$ of our signal is generated according to the conditional density $f(X_n|\widetilde{X}_n)$, where the conditioning is on the value of the corresponding entry of the SI, meaning the signal $x$ now has independent, but not identically distributed, entries.  Owing to $x$ no longer being i.i.d., the conditional denoiser \eqref{eq:eta_CAMP} depends on the index $n$, meaning that different scalar denoisers will be used at different indices, based on different SI at different indices.  Both results~\cite{BayMont11} and ~\cite{RushV16} require that the same denoiser function be applied to each element of the pseudo-data and our denoiser will change element-wise based on the SI.

Recent results~\cite{berthier2017} extend the asymptotic SE analysis to a larger class of possible denoisers, allowing, for example, each element of the input to use a different non-linearity as is the case in AMP-SI.  We employ these results to rigorously relate the SE presented in \eqref{eq:SE2} to the AMP algorithm in \eqref{eq:CAMP1} - \eqref{eq:CAMP2} when considering the $\ell_2-$error between the pseudo-data  and the true signal.  To do so, we make the following assumptions:
\textbf{(A1)} The measurement matrix $A$ has i.i.d.\ mean-zero, Gaussian entries having variance $1/M$.
\textbf{(A2)}  The noise $z$ is i.i.d.\ $\sim f(Z)$ with finite variance $\sigma_z^2$.
\textbf{(A3)} The signal and SI $(X, \widetilde{X})$ are sampled i.i.d.\ from the joint density $f(X, \widetilde{X})$.
\textbf{(A4)} 
For $t\geq 0$, the denoisers $\eta_t: \mathbb{R}^{2} \rightarrow \mathbb{R}$ defined in \eqref{eq:eta_CAMP} are Lipschitz in their first argument.
\textbf{(A5)} For  any $2 \times 2$ covariance matrix $\Sigma$, let $(Z_1, Z'_1) \sim \mathcal{N}(0,\Sigma)$  independent of $(X, \widetilde{X}) \sim f(X, \widetilde{X}).$  Then for any $s, t \geq 0$,
\[\mathbb{E} [X \eta_t(X + Z_1, \widetilde{X})] < \infty,\]
and
\[\mathbb{E} [\eta_t(X + Z_1, \widetilde{X}) \eta_s(X + Z'_1, \widetilde{X}) ] < \infty.\]

Under the above assumptions we have the following guarantee relating the SE to the $\ell_2-$error.
\begin{theorem}
\label{thm:SE}
 Under assumptions $\textbf{(A1)} - \textbf{(A5)}$,
\[\lim_{N \rightarrow \infty}  \frac{1}{N} ||x^t + A^T r^t- x||^2 \overset{p}{=}  \lambda_t^2,\]
where $x^t$ and $r^t$ are iterates of AMP-SI as shown in \eqref{eq:CAMP1}-\eqref{eq:CAMP2} and $\overset{p}{=}$ indicates convergence in probability.
\end{theorem}
\begin{proof}
The proof uses \cite[Thm.\ 14]{berthier2017}.  We note that conditions $(C5)$ and $(C6)$ needed in \cite[Thm.\ 14]{berthier2017} follow from our assumptions $\textbf{(A2)}$ and $\textbf{(A5)}$ and the Law of Large Numbers.
\end{proof}
In Appendix~\ref{app:SEtheory} we show an example of how to verify the assumptions $\textbf{(A4)}$ and  $\textbf{(A5)}$ for a simple case where the signal is i.i.d.\ Gaussian and the SI is the signal plus i.i.d.\ Gaussian noise.

Ongoing work considers extensions of Theorem~\ref{thm:SE} to more general loss functions and weaker assumptions than those made in $\textbf{(A1)} - \textbf{(A5)}$.  We believe that by using the theory supporting non-separable denoisers provided in~\cite{berthier2017}, it is possible to extend our AMP-SI framework to handle arbitrary joint distributions between the signal and SI (with element-wise dependencies) and that it is possible to extend the framework to the vector AMP algorithm~\cite{RanganVAMP, VAMP_NIPS} allowing for a more general class of measurement matrices.

\subsubsection{Bayes Optimality}
When the conditional expectation denoiser~\eqref{eq:cond_dist_denoiser} is used in AMP~\eqref{eq:AMP1}-\eqref{eq:AMP2}, the corresponding SE~\eqref{eq:AMP_SE} in its convergent states coincides with Tanaka's fixed point equation~\cite{Tanaka2002,GuoVerdu2005}, ensuring that if AMP runs until it converges, in the large system limit the result provides the best possible MSE achieved by any algorithm under certain problem conditions. Tanaka's fixed point equation in the Gaussian case has been rigorously proven, see \cite{reeves2016replica,barbier2017mutual}.

In the case that the SI available to the system is a Gaussian-noise corrupted view of the true signal,
i.e., $\widetilde{X} = X + \mathcal{N}(0, \widehat{\sigma}^2)$,
it can be shown~\cite{BNMRW2017} that the fixed points of AMP-SI SE \eqref{eq:SE2} coincide with the fixed points of AMP SE \eqref{eq:AMP_SE} with `effective' measurement rate
$\delta_{eff} = \delta/\mu$ and
`effective' measurement noise variance $\sigma_{eff}^2 = \mu \sigma^2 $ where $0 \leq \mu \leq 1$ and the
$\mu$ depends on the prior density of the signal and the SI noise variance $ \widehat{\sigma}^2$. The
effective change in $\delta$ and $\sigma^2$
implies that the incorporation of Gaussian-noise corrupted SI via the AMP-SI algorithm gives us Bayes-optimal signal recovery for a standard (without SI) linear regression problem \eqref{eq:hdreg} with \emph{more} measurements and \emph{reduced} measurement noise variance than our own.  The details of this argument are provided in Appendix~\ref{app:asilomar} and first
appeared in~\cite{BNMRW2017}.
We believe AMP-SI has similar Bayes-optimality properties to standard AMP, however, proving this rigorously is theoretically difficult since the above analysis relies heavily on the Gaussianity of the SI noise, and thus may be difficult to generalize.

\section{Bernoulli-Gaussian Model} \label{sec:BG}
The BG model reflects the scenario in which one wants to recover a sparse signal and has access to SI in the form of the signal with additive white Gaussian noise (AWGN). In other words, at every iteration the algorithm has access to SI, $\widetilde{x},$ and pseudo-data, $v^t$, with
\begin{align*}
\widetilde{x} &= x + \mathcal{N}(0,\widehat{\sigma}^2 \mathbb{I}), \qquad v_t \approx x + \mathcal{N}(0, \lambda^2_t \mathbb{I}),
\end{align*}
where the additive noise in the SI and pseudo-data are independent. The entries of 
$x$ follow a BG pdf:
\begin{align}
X_n \sim
\epsilon \frac{1}{\sqrt{2\pi}} \exp\left(\frac{-x_n^2}{2}\right) + (1-\epsilon)\delta_0,
\label{eq:BG}
\end{align}
so that $x$ is zero with probability $1-\epsilon$ and is standard Gaussian in nonzero entries. Here, $\delta_0$ represents the Dirac delta function at 0. We start with this model because it provides a closed form derivation of the denoiser with an intuitive interpretation. The next section will show that even for this toy model, the derivation is not trivial.

\subsection{The Conditional Denoiser with SI for BG}
In this section we will derive the following result:

\begin{result} The AMP-SI denoiser \eqref{eq:eta_CAMP}
has the following closed form for the BG model:
 \begin{align}
 \eta( a,b ) &= \Big(1 + R_{(a,b)}\Big)^{-1} \left[ \frac{a \widehat{\sigma}^2 + b \lambda_t^2}{ \widehat{\sigma}^2 + \lambda_t^2 + \widehat{\sigma}^2 \lambda_t^2 } \right],
 \label{eq:denoiser_BG}
 \end{align}
 where $R_{(a,b)}$ is a ratio between probabilities (computed in \eqref{eq:S-bg}), $\widehat{\sigma}^2$ is the
variance of the AWGN of the SI, and $\lambda_t^2$ is the
variance of the AWGN of the pseudo-data at iteration $t$.
 \label{result1}
 \end{result}
Note that the denoiser given in~\eqref{eq:denoiser_BG} behaves as we would expect as the parameters in the problem change.  Specifically, by considering the definition of $R_{(a,b)}$ in \eqref{eq:S-bg}, we can see that the term $(1 + R_{(a,b)})^{-1}$ approaches $1$ as the BG sparsity parameter, $\epsilon$, approaches $1$, and approaches $0$ as $\epsilon$ approaches $0$, meaning that as the signal gets more sparse ($\epsilon$ approaches $0$) the denoiser provides more shrinkage ($(1 + R_{(a,b)})^{-1}$ approaches $0$). The second term of \eqref{eq:denoiser_BG}, i.e.\ $ \frac{a \widehat{\sigma}^2 + b \lambda_t^2}{ \widehat{\sigma}^2 + \lambda_t^2 + \widehat{\sigma}^2 \lambda_t^2 }$,  is a weighted sum between the pseudo-data and the SI. As the SI noise $\widehat{\sigma}^2$ increases, a larger weight is placed on the pseudo-data.  Similarly as the noise in the pseudo-data $\lambda_t^2$ increases, a larger weight is placed on the SI.

Now we derive Result~\ref{result1}.  In what follows, the notation $\psi_{\tau^2}(x)$ refers to the zero-mean Gaussian density with variance $\tau^2$ evaluated at $x$. We will use $f(\cdot)$ (or $f(\cdot, \cdot)$, $f(\cdot, \cdot,\cdot)$, and so on)
to represent a generic pdf (or joint pdf) on the input. Before we begin the derivation, we introduce a few lemmas relating to computations involving two RVs $A = \rho X + \mathcal{N}(0, \sigma_a^2)$ and $B = X + \mathcal{N}(0,\sigma_b^2)$. Deriving the conditional denoiser for BG (and later BDD) requires the joint pdf between $A$ and $B$ (Lemma \ref{lem:jointGauss}),
the product of two Gaussian pdfs (Lemma \ref{lem:gauss_prod}),
and the expectation of $X$ conditional on instances of $A$ and $B$ (Lemma~\ref{lem:jointCond}).

\begin{lemma} Given instances $a$ and $b$ such that $A = \rho X + \mathcal{N}(0, \sigma_a^2)$ for some constant $\rho$, $B = X + \mathcal{N}(0,\sigma_b^2)$, and $X \sim \mathcal{N}(0, \sigma_x^2)$, the joint pdf between $A$ and $B$ is:
\begin{align*}
f(a, b) &= \frac{1}{\rho} \psi_{\sigma_x^2 + \sigma_b^2}(b) \psi_{\frac{\sigma_x^2 \sigma_b^2}{\sigma_x^2 + \sigma_b^2} + \frac{\sigma_a^2}{\rho^2}} \left( \frac{\sigma_x^2 b}{\sigma_x^2+\sigma_b^2}  - \frac{a}{\rho} \right),
\end{align*}
assuming that the AWGN in $A$, AWGN in $B$, and $X$ are independent.
\label{lem:jointGauss}
\end{lemma}

Lemma \ref{lem:jointGauss} is proved in Appendix~\ref{app:jointGauss}.

Below, we denote the $\mathcal{N}(\mu, \sigma^2)$ density evaluated at $x$ by $\widetilde{\psi}_{\mu, \sigma^2}(x)$.

The next lemma provides a simplified expression for the product of two Gaussian densities.
\begin{lemma}
For two Gaussian densities,
$\widetilde{\psi}_{\mu_1, \sigma_1^2}(x) \times \widetilde{\psi}_{\mu_2, \sigma_2^2}(x)$ equals
\[\widetilde{\psi}_{\left( \frac{ \mu_1 \sigma_2^2 +  \mu_2 \sigma_1^2}{\sigma_1^2 + \sigma_2^2}, \frac{\sigma_1^2 \sigma_2^2}{\sigma_1^2 + \sigma_2^2}\right)}(x) \times \widetilde{\psi}_{(\mu_1 - \mu_2, \sigma_1^2 + \sigma_2^2)}(0).\]
\label{lem:gauss_prod}
\end{lemma}
The proof of Lemma \ref{lem:gauss_prod} involves
straightforward algebra and completing the square;
the lemma could also be formulated as
a convolution of three Gaussian densities.

The final lemma generalizes the conditional expectation of a Gaussian random variable $X$ conditioned on the value of two noisy versions of $X$, particularly $A \sim \rho X + \mathcal{N}(0, \sigma_a^2)$ and $B \sim X + \mathcal{N}(0,\sigma_b^2)$. We will use the shorthand notation $\mathbb{E}[X \,|\, a,b]$ to mean
\[
\mathbb{E} [X | a =  \rho X + \mathcal{N}(0, \sigma_a^2), b = X + \mathcal{N}(0,\sigma_b^2)].
\]
\begin{lemma} The conditional expectation of a Gaussian RV $X \sim \mathcal{N}(0, \sigma_x^2)$ given instances $a$ and $b$ such that $A \sim \rho X + \mathcal{N}(0, \sigma_a^2)$ for some constant $\rho$ and $B \sim X + \mathcal{N}(0,\sigma_b^2)$ can be computed as:
\begin{equation*}
\mathbb{E} [X \,|\, a,b] = \frac{ \rho\sigma_x^2 \,  \sigma_b^2 a +  \sigma_x^2 \sigma_a^2 b}{\sigma_x^2 \, (\sigma_a^2 + \rho^2 \sigma_b^2) + \sigma_a^2 \sigma_b^2},
\end{equation*}
assuming that the AWGN in $A$, AWGN in $B$, and $X$ are independent.
\label{lem:jointCond}
\end{lemma}
The proof of Lemma \ref{lem:jointCond} can be found in Appendix~\ref{app:jointCond}.

\subsection{Derivation of the Denoiser with SI for BG}
Using the aforementioned lemmas, we derive the conditional denoiser for the BG model.

{\bf Derivation of Result \ref{result1}}.  To derive Result \ref{result1}, note that
\[
{\small
\eta(a,b) = \mathbb{E} [X | a = X + \mathcal{N}(0, \lambda_t^2), b =  X + \mathcal{N}(0,\widehat{\sigma}^2)],
}
\]
and therefore,
\begin{align}
\eta(a,b)
 & =  \Pr(X \neq 0 \, | \, a,b ) \mathbb{E} [ X \, | \, a,b, X \neq 0 ]. \label{eq:eta-bg}
 \end{align}
Simplifying the expression $\Pr( X \neq 0 \, | \, a,b )$,
\begin{align}
\Pr(X \neq 0 \, | \,  a,b)
&= \frac{f(X \neq 0, a, b)}{f(X \neq 0, a, b) + f(X = 0, a,b)}
\label{eq:1-bg} \\
&= \left[1 + \frac{\Pr(X=0) f(a,b \, | \, X = 0)}{\Pr(X \neq 0) f(a, b \, | \, X \neq 0)}\right]^{-1}. \nonumber
\end{align}
Note that here we slightly abuse the notation of a pdf with an event (i.e., $X \neq 0$ or $X=0$) as an input to the density function.
Considering the ratio in \eqref{eq:1-bg}, define 
\[R_{(a,b)} =\frac{\Pr(X=0) f(a,b \, | \, X = 0)}{\Pr(X \neq 0) f(a, b \, | \, X \neq 0)}.\]
Conditioned on $X \neq 0$, we can compute $ f(a, b \, | \, X \neq 0)$ using Lemma \ref{lem:jointGauss} with $\rho = 1$, $\sigma_x^2 =1$, $\sigma_a^2 = \lambda_t^2$, and $\sigma_b^2 = \widehat{\sigma}^2$:
\begin{align*}
f(a,b| X \neq 0) =  \psi_{1 + \widehat{\sigma}^2}(b) \psi_{\frac{\widehat{\sigma}^2}{1 + \widehat{\sigma}^2} +\lambda_t^2} \left( \frac{b}{1+\widehat{\sigma}^2}  - a \right).
\end{align*}
Also, when $X = 0$, $A$ and $B$ are independent so
\begin{align*}
f(a,b \, | \, X = 0) &= f(a \, | \, X = 0) f(b \, | \, X = 0) \\
& = \psi_{\lambda_t^2}(a) \psi_{\widehat{\sigma}^2}(b).
\end{align*}
With these elements, we can compute $R_{(a,b)}$:
\begin{equation}
R_{(a,b)} = \frac{(1-\epsilon) \psi_{\lambda_t^2}(a) \psi_{\widehat{\sigma}^2}(b)}{\epsilon \psi_{1 + \widehat{\sigma}^2}(b) \psi_{\frac{\widehat{\sigma}^2}{1 + \widehat{\sigma}^2} +\lambda_t^2} \left( \frac{1}{1+\widehat{\sigma}^2} b - a \right)}.
\label{eq:S-bg}
\end{equation}
The last term we must compute is the conditional expectation in \eqref{eq:eta-bg}. Using Lemma \ref{lem:jointCond} with $\rho = 1$, $\sigma_x^2 =1$, $\sigma_a^2 = \lambda_t^2$, and $\sigma_b^2 = \widehat{\sigma}^2$, we have that
\begin{equation}
\mathbb{E}[X | a,b] = \mathbb{E}[X | a,b, X \neq 0] = \frac{\widehat{\sigma}^2 a +  \lambda_t^2 b}{ \lambda_t^2 + \widehat{\sigma}^2 + \lambda_t^2 \sigma_b^2}.
\label{eq:exp-bg}
\end{equation}

Result \ref{result1} is obtained by combining the above computations. In particular, we have that
\begin{equation*}
\eta(a,b) = \left(1 + R_{(a,b)}\right)^{-1} \mathbb{E} [X | a,b],
\end{equation*}
where $R_{(a,b)}$ and $ \mathbb{E} [X  | a,b]$ are computed in \eqref{eq:S-bg} and \eqref{eq:exp-bg}, respectively.

\subsection{State Evolution for BG}
Using the denoiser in \eqref{eq:denoiser_BG}, we can compute the SE equations \eqref{eq:SE2}.  Letting $\delta = \frac{M}{N}$, we have $\lambda_0^2 =  \sigma_z^2 + \frac{1}{\delta}\mathbb{E}[X^2]$ and for $t \geq 0$,
 \begin{align*}
\lambda_{t+1}^2 =  \sigma_z^2 + \frac{1}{\delta} \mathbb{E}\left[ (\eta_t(X+\lambda_t Z_{1}, X+\widehat{\sigma} Z_{2})-X)^2\right],
\end{align*}
where $\eta_t(\cdot, \cdot)$ is defined in \eqref{eq:denoiser_BG}, $Z_1$ and $Z_2$ are independent, standard Gaussian RVS that are independent of $X \sim f(X)$, and the expectation is with respect to $Z_1, Z_2$, and $X$. Because the form of the derived denoiser is complicated, it seems infeasible to find a closed-form expression for $\lambda_{t+1}^2$. Instead, we approximate the SE in our numerical experiments.

\section{Birth-Death-Drift Model} \label{sec:BDD}
In this section, we investigate the application of AMP-SI on a stochastic signal model closely resembling the channel estimation problem in wireless communications.
Our birth-death-drift (BDD) model is based on Markov processes for the supports and
amplitudes of signal elements, and has been studied in the time-varying signal literature~\cite{DCSAMP}.
This section presents the dynamics of individual elements more explicitly.

\subsection{Connections to Channel Estimation}
\label{subsec:channel_est}

\textbf{BDD Motivation:} Our channel estimation scenario is illustrated in Fig.~\ref{fig:batch}.
Typical wireless devices transmit a pilot sequence and data payload in {\em batches}.
In batch~$b$, the pilot sequence $p$ is transmitted into the channel, where it is convolved
with the channel response $x^b$, yielding noisy linear measurements,
\[
y^b = \text{conv}(p,x^b) + z.
\]
This convolution, $\text{conv}(\cdot,\cdot)$,
can be expressed as the product of a Toeplitz matrix with a vector,
\begin{equation}
\label{eq:Toeplitz}
y^b = \text{Toeplitz}(p)x^b + z,
\end{equation}
where $\text{Toeplitz}(p)$ is the Toeplitz matrix that corresponds to the
pilot sequence $p$.
To perform channel estimation using AMP-SI, we will consider (\ref{eq:Toeplitz})
as a linear inverse problem (\ref{eq:hdreg}), where
$\text{Toeplitz}(p)$ is the measurement matrix.
Our goal will be to estimate the channel response $x^b$ in batch~$b$ using
the noisy measurements $y^b$, matrix $\text{Toeplitz}(p)$, and
$\widetilde{x} = \widehat{x}^{b-1}$, our estimate of the channel response in the previous batch, $b-1$
(Fig.~\ref{fig:batch}). Our resulting estimate for the channel response, $\widehat{x}^b$,
will then help us estimate the channel response in the next batch, $x^{b+1}$.
To develop a conditional denoiser, we need a channel model that describes the
channel response $x^b$, and especially its dependence on $x^{b-1}$,
the channel response in the previous batch.
We model the channel as an (unknown) finite impulse response (FIR) filter,
whose taps correspond to the amplitude of the channel response at different delays.
Many filter taps are close to zero, and this sparsity makes the channel estimation problem
a sparse signal recovery task.

Due to the slowly varying time dynamics of the channel, $x^b$ is not only sparse,
but has strong dependencies with the channel response in adjacent batches.
A possible model for changes from $x^b$ to $x^{b+1}$ involves
({\em i}) {\em birth} of new nonzeros in $x^{b+1}$ (corresponding to new wireless paths);
({\em ii}) {\em death} of nonzeros in $x^b$ that become zero in
$x^{b+1}$ (existing paths are obscured as the user moves); and
({\em iii}) slow {\em drift} of existing nonzeros.
We call these time-varying channel dynamics a
{\em birth-death-drift} (BDD) model.
To demonstrate the efficacy of our BDD model, we looked at ray tracing simulations
for a mobile user moving in an urban environment. A photo of the urban environment (a suburb of Washington, DC) is shown in the top panel of Fig.~\ref{fig:raytrace}.
The bottom panel shows two realizations of the channel filter. The realization corresponding to the beginning of the mobile user's motion is depicted by circles,
and the realization corresponding to the end of the user's motion is marked by squares.
It can be seen that most nonzero taps of the channel filter drift slowly; birth and death events are highlighted for the reader's convenience.
Not only is the channel filter in each batch sparse,
but its differences relative to filters in previous batches are highly structured.

\begin{figure}[t!]
\centering
\includegraphics[width=.45\textwidth]{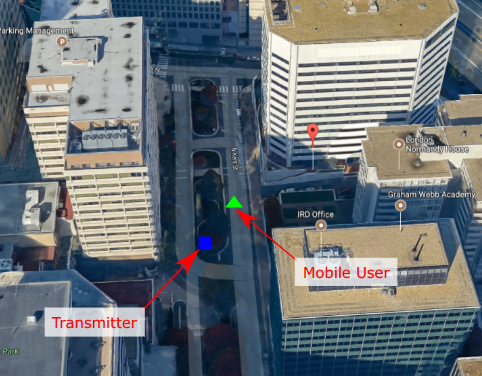}
\includegraphics[width=.5\textwidth]{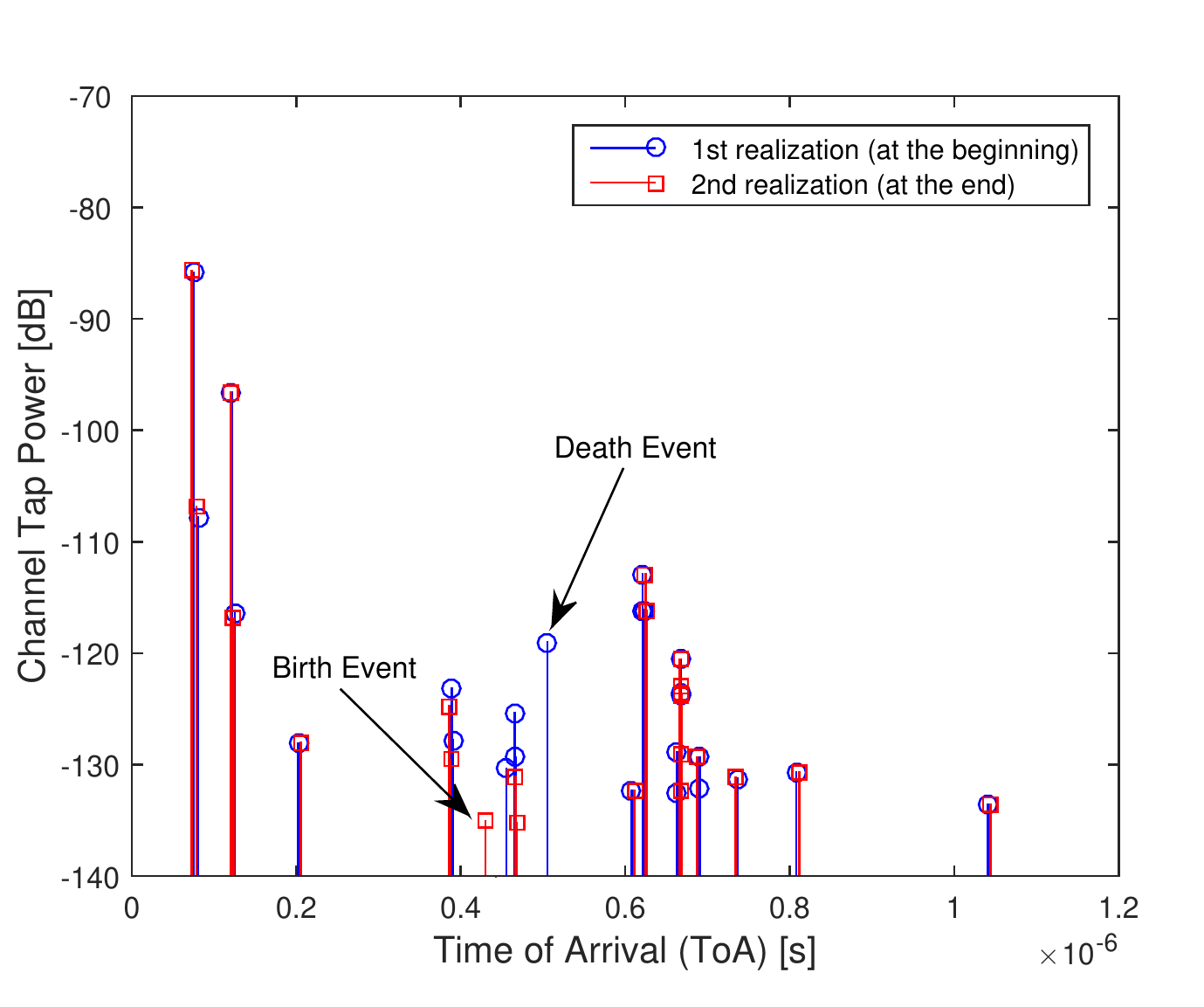}
\caption{Ray tracing simulation results for a mobile user moving in an urban environment (top) show that the channel realizations at the beginning and end help (bottom) lend credence to our BDD model.}\label{fig:raytrace}
\vspace*{-4.5mm}
\end{figure}

For communication-minded readers, the proposed BDD model resembles that of Saleh and Valenzuela~\cite{saleh1987statistical}.
Our paper uses the BDD model for filter taps that are independent within each batch, meaning there are no inter-batch dependencies.
Slow temporal dynamics over multiple batches are prevalent in
real-world channels.
For example, in typical wireless channels a death process involves
the path energy shrinking gradually over multiple time batches~\cite{MacCartney_2017}.
However, our BDD model does not support such dynamics.
Although we only demonstrate the efficacy of AMP-SI on the simplified BDD model,
the framework can be adapted to other settings with SI.
For example,
BDD could be extended to assign different variances to the Gaussian distributions associated with different taps based on a predefined power delay profile (PDP).

{\bf Formal definition of BDD model:} To formally introduce the BDD model, we start by considering a single time batch. Between the previous and current batch, the signal elements independently change according to a BDD
process which defines the joint pdf $f(X_{p}, X_{c})$, where `p' denotes the previous signal, a noisy version that serves as SI, and `c' the current signal.

 The elements of the signal evolve following four cases in the BDD model: for any entry $n \in 1, 2, \ldots, N$,\\
 {\bf Case\hspace*{1mm}1:}\ Zero entry remains zero, i.e.,  $[x_p]_n=0$ and $[x_c]_n=0$. \\
{\bf Case\hspace*{1mm}2:}\ {\em Death} -- nonzero entry becomes zero, i.e.,  $[x_p]_n \sim\mathcal{N}(0,\sigma^2_s)$ and $[x_c]_n =0$. \\
{\bf Case\hspace*{1mm}3:}\ {\em Drift} -- nonzero entry remains nonzero, i.e.,  $[x_p]_n \sim\mathcal{N}(0,\sigma^2_s)$ and $[x_c]_n = \rho [x_p]_n +\mathcal{N}(0,\sigma^2)$. \\
{\bf Case\hspace*{1mm}4:}\ {\em Birth} -- zero entry becomes nonzero, i.e.,  $[x_p]_n =0$ and $[x_c]_n \sim\mathcal{N}(0,\sigma_s^2)$.\\
We define $\sigma^2 > 0$ to be the variance in the zero-mean Gaussian drift and $\sigma_s^2 > 0$ to be the steady-state variance, or the variance of the nonzero entries in the signal at every batch. Indeed, an entry of
the current signal is nonzero in Cases~3 and~4,
and by choosing the constant $\rho>0$ such that $\rho^2 \sigma^2_s + \sigma^2=\sigma^2_s$ , we ensure $var(X_c)= \sigma^2_s$ for both these cases.
Finally, Case j occurs with probability $\epsilon_j$
and $\sum_{j=1}^4 \epsilon_j = 1$.

\begin{remark} The BG model is a simplified version of the BDD model. One can confirm that setting $\epsilon_2 = \epsilon_4 = 0$, $\epsilon_1 = 1-\epsilon$, $\epsilon_3 = \epsilon$, $\sigma = 0$, and $\sigma_s^2 = 1$ obtains the model discussed in Section~\ref{sec:BG}.
\end{remark}

\begin{remark} In Case~3, a zero-mean Gaussian random variable with variance $\sigma^2$
represents short-term fading due to multipath and oscillator drift.
Similarly, $\rho$ represents inter-batch correlations or drift
between nonzero elements of $x$,
and is inversely correlated to the amount of fading in a wireless channel \cite{FadeShadow}.
\end{remark}

In the BDD model, the SI takes the form of the previous batch's signal $x_p$ with AWGN. The pseudo-data, which we label $v_t$, is approximately the current batch's signal $x_c$ with AWGN. That is, at every iteration the algorithm has access to:
\begin{align*}
\widetilde{x} &= x_p + \mathcal{N}(0,\widehat{\sigma}^2 \mathbb{I}), \qquad v_t \approx x_c + \mathcal{N}(0, \lambda^2_t \mathbb{I}),
\end{align*}
where the additive noise in the SI and pseudo-data are independent.
In the multiple batch setting, the pseudo-data in the final iteration of AMP-SI for approximating the $b^{th}$ signal, which is a noisy version of $x_p$, becomes the SI for the approximation of the $(b+1)^{th}$ signal and the variance of this SI is available through $\lambda_t^2$ given by the SE equations~\eqref{eq:SE2}.

\subsection{The Conditional Denoiser with SI for BDD}

We now derive the conditional denoiser for the BDD model presented in Section~\ref{subsec:channel_est}.
Recall that the inputs $a$ and $b$ of the conditional denoiser $\eta(a,b)$ are instances of the
pseudo-data $v_t$ and SI $\widetilde{x}$, respectively.

\begin{result} The AMP-SI denoiser \eqref{eq:eta_CAMP}
has the following closed form for the BDD model,
 \begin{align}
 \eta( a,b )  &= \frac{ \epsilon_4 \, \mu^4_{(a,b)}}{S_{(a,b)}}  \left[\frac{\sigma^2_s \, a}{\sigma^2_s + \lambda^2_t}\right]  \label{eq:denoiser} \\
&\quad + \frac{ \epsilon_3 \, \mu^3_{(a,b)} }{S_{(a,b)}}\left[ \frac{\sigma_s^2 \, (\sigma^2 + \widehat{\sigma}^2) \, a + \rho  \, \sigma_s^2 \,  \lambda_t^2 \, b}{\sigma_s^2 \, (\sigma^2 + \lambda_t^2 + \widehat{\sigma}^2) + \lambda_t^2 \, \widehat{\sigma}^2}  \right], \nonumber
 \end{align}
where $\epsilon_i \mu^i_{(a,b)}$ is the the joint pdf evaluated for Case i and instances $a$ and $b$. Additionally, $S_{(a,b)}$ is the marginal pdf evaluated at instances $a$ and $b$. The variables $\mu^3_{(a,b)}$, $\mu^4_{(a,b)}$, and $S_{(a,b)}$ are defined in \eqref{eq:S} below.
\label{result2}
\end{result}

In what follows, the notation $\psi_{\tau^2}(x)$ refers to the zero-mean Gaussian density with variance $\tau^2$ evaluated at $x$.
\begin{equation}
\begin{split}
\mu^3_{(a,b)} &=  \psi_{\frac{\sigma^2_s \, ( \widehat{\sigma}^2 +  \sigma^2)}{ \widehat{\sigma}^2 + \sigma^2_s} +\lambda_t^2} \, \left(\frac{ \rho  \, \sigma^2_s \, b}{ \widehat{\sigma}^2 + \sigma^2_s} - a\right) \, \psi_{\widehat{\sigma}^2 + \sigma^2_s}(b), \\
\mu^4_{(a,b)} &=   \psi_{\sigma^2_s +\lambda_t^2}(a) \, \psi_{ \widehat{\sigma}^2}(b) , \\
S_{(a,b)} &= \epsilon_1  \,  \psi_{\lambda_t^2}(a)\, \psi_{ \widehat{\sigma}^2}(b) +\epsilon_2  \, \psi_{\lambda_t^2}(a) \, \psi_{ \widehat{\sigma}^2 + \sigma_s^2}(b) \\
&\quad + \epsilon_3 \, \mu^3_{(a,b)} + \epsilon_4 \, \mu^4_{(a, b)}.
\label{eq:S}
\end{split}
\end{equation}

\subsection{Derivation of the Denoiser for BDD}

Using the lemmas presented in Section~\ref{sec:BG}, we derive the conditional denoiser for the BDD model.

{\bf Derivation of Result~\ref{result2}.} To derive Result~\ref{result2}, note that
{\small
\[
\eta(a,b) = \mathbb{E} [X_c | a = X_c + \mathcal{N}(0, \lambda_t^2), b =  X_p + \mathcal{N}(0,\widehat{\sigma}^2)],
\]}
which we represent with shorthand $\mathbb{E}[X_c|a,b]$. Then,
\begin{align}
 \eta(a,b)
 & = \sum_{j=3}^4 \Pr( \text{Case j} \, | \, a,b ) \mathbb{E}[ X_c | a,b, \text{Case j} ], \label{eq:etasum}
 \end{align}
where we use the fact that $x_c=0$
in Cases~1 and~2, and so
$\mathbb{E}[X_c \, |\, a, b, \text{Case 1}]= \mathbb{E}[X_c | a, b, \text{Case 2}]=0$.
Considering \eqref{eq:etasum}, let us simplify the expression $\Pr( \text{Case j} \, | \, a,b )$.  In the following we use $f(\cdot)$ (or $f(\cdot, \cdot)$, $f(\cdot, \cdot,\cdot)$, and so on)
to represent a generic pdf (or joint pdf) on the input. By Bayes' Rule,
\begin{align}
&\Pr(\text{Case j} \, | \,  a,b)
= \frac{f(\text{Case j},a,b)}{\sum_{i = 1}^4 f(\text{Case i}, a,b)}.  \label{eq:1}
\end{align}

To derive the denoiser \eqref{eq:denoiser} from \eqref{eq:etasum} and \eqref{eq:1},
we must compute, for $j = \{1, 2, 3, 4\}$:
\begin{equation}
\label{eq:f(Casej,xhat,vt)}
f(\text{Case j}, a, b )=\Pr(\text{Case j}) f( b| \text{Case j}) f( a | \text{Case j},  b),
\end{equation}
along with $\mathbb{E}[ X_c \, |  \, a, b, \text{Case 3} ]$ and $\mathbb{E}[ X_c \, |  \, a, b, \text{Case 4} ]$.

We first address Cases 1, 2, and 4 since
$a = X_c + \mathcal{N}(0, \lambda_t^2)$ and $b =  X_p + \mathcal{N}(0,\widehat{\sigma}^2)$ are
independent in these cases.
In Case 3, these values are {\em dependent} and therefore that case is handled carefully at the end.

\textbf{Cases 1, 2, and 4:} Here, we can simplify \eqref{eq:f(Casej,xhat,vt)} by noting that $f(a \, |  \, \text{Case j}, b) = f(a \, |  \, \text{Case j})$ due to the independence of $a$ and $b$ in these cases. For $j \in \{1, 2, 4 \}$,
\begin{align}
f( \text{Case j}, a,b ) &=  \Pr(\text{Case j}) f(b \,  | \,  \text{Case j}) f(a  \, |  \, \text{Case j}) \nonumber\\
&= \epsilon_j \, \psi_{ \sigma_{b,j}^2}(b) \, \psi_{\sigma_{a,j}^2}(a), \label{eq:cj-prob}
\end{align}
where $\sigma_{a,j}^2 = \mathbb{E}[a^2 \, | \, \text{Case j}]$, and $\sigma_{b,j}^2 = \mathbb{E}[b^2 \, | \,\text{Case j}]$. We also compute $\mathbb{E}[ X_c \, | \, a, b, \text{Case 4} ]$.  This equals $\mathbb{E}[ X_c \, | \, a,  \text{Case 4} ] $ since $b =  \mathcal{N}(0,\widehat{\sigma}^2)$
is independent of $X_c$.
Since $a = X_c + \mathcal{N}(0, \lambda_t^2)$,
the conditional expectation is computed using a Wiener filter,
\begin{align}
\mathbb{E}[ X_c \, |  \, a,  \text{Case 4} ] &= \mathbb{E}[ X_c \, |  \, X_c + \mathcal{N}(0, \lambda_t^2)]= \frac{\sigma^2_s a}{\sigma^2_s + \lambda^2_t}. \label{eq:c4-exp}
\end{align}

\textbf{Case 3:} Here, $a = \rho X_p + \mathcal{N}(0, \sigma^2) + \mathcal{N}(0, \lambda_t^2)$ and $b =  X_p + \mathcal{N}(0,\widehat{\sigma}^2)$
which, in contrast to the above cases, are now dependent through $X_p \sim \mathcal{N}(0, \sigma_s^2)$.
To compute $f(\text{Case 3}, a, b) = P(\text{Case 3})f(a,b|\text{Case 3})$ note that conditional on Case 3, we may apply Lemma \ref{lem:jointGauss} to $f(a,b| \text{Case 3})$ with $X = X_p$, $\sigma_a^2 = \sigma^2 + \lambda_t^2$, and $\sigma_b^2 = \widehat{\sigma}^2$ to obtain:
\begin{align}
&f( \text{Case 3}, a, b) =  \Pr(\text{Case 3}) \, f(a,b  | \text{Case 3}) \nonumber \\
&=  \frac{ \epsilon_3}{\rho} \psi_{\sigma^2_s + \widehat{\sigma}^2}(b) \, \psi_{\frac{\sigma^2_s ( \widehat{\sigma}^2 +  \sigma^2)}{ \sigma^2_s + \widehat{\sigma}^2} + \frac{\sigma^2 + \lambda_t^2}{\rho^2}}\left(\frac{  \sigma^2_s b }{ \sigma^2_s + \widehat{\sigma}^2} - \frac{a}{\rho}\right).\label{eq:c3-prob}
\end{align}

We also need to compute $\mathbb{E}[ X_c | a, b, \text{Case 3} ]$. By linearity of expectation we have
\begin{align}
&\mathbb{E}[ X_c \,| \, a, b, \text{Case 3} ] = \mathbb{E}[ \rho X_p + \mathcal{N}(0, \sigma^2) \, | \, a, b, \text{Case 3}] \nonumber \\
&=  \rho \mathbb{E}[ X_p \, | \, a, b, \text{Case 3} ] +  \mathbb{E}[  \mathcal{N}(0, \sigma^2) \, | \, a, b, \text{Case 3}]. \label{eq:twoterms}
\end{align}
Conditional on Case 3, we can compute the first expectation in \eqref{eq:twoterms} using Lemma \ref{lem:jointCond} with $X = X_p$, $\sigma_a^2 = \sigma^2 + \lambda_t^2$ since $a = \rho X_p + \mathcal{N}(0, \sigma^2 + \lambda_t^2)$, and $\sigma_b^2 = \widehat{\sigma}^2$ since $b = X_p + \mathcal{N}(0, \widehat{\sigma}^2)$:
{\small
\begin{align}
\mathbb{E}[ X_p  | a, b, \text{Case 3} ] &=  \frac{\sigma_s^2 [\rho  \, \widehat{\sigma}^2 a +  (\sigma^2 + \lambda_t^2) b]}{\sigma_s^2  (\sigma^2 + \lambda_t^2 + \rho^2 \widehat{\sigma}^2) + (\sigma^2 + \lambda_t^2) \widehat{\sigma}^2} \nonumber \\
&= \frac{\sigma_s^2[\rho  \,  \widehat{\sigma}^2 a +   (\sigma^2 + \lambda_t^2) b]}{\sigma_s^2(\sigma^2 + \lambda_t^2 + \widehat{\sigma}^2) + \lambda_t^2 \widehat{\sigma}^2}, \label{eq:c3-exp3a}
\end{align}}
\hspace{-4pt}where we use the fact that $\rho^2 \sigma_s^2 + \sigma^2 = \sigma_s^2$ to simplify. Letting $Z_c \sim \mathcal{N}(0, \sigma^2)$ be such that $X_c = \rho X_p + Z_c$, one can use the same approach as in Lemma \ref{lem:jointCond} to obtain:
\begin{align}
&\mathbb{E}[ Z_c  \, | \, a = Z_c + \rho X_p + \mathcal{N}(0, \lambda^2), \\
& \qquad  \qquad b = X_p + \mathcal{N}(0, \widehat{\sigma}^2), \text{Case 3} ] \nonumber \\
&= [\sigma^2 \ \ 0]
\left[ {\begin{array}{cc}
\rho^2 \sigma^2_s + \lambda^2 & \rho \sigma^2_s\\
\rho\sigma^2_s & \sigma^2_s+\widehat{\sigma}^2
\end{array} } \right]^{-1}
\begin{bmatrix} a \\ b \end{bmatrix} \nonumber \\
&= \frac{\sigma^2[(\sigma_s^2 + \widehat{\sigma}^2)a - \rho \sigma_s^2 b] }{\sigma_s^2(\sigma^2 + \lambda_t^2 + \widehat{\sigma}^2) + \lambda_t^2 \widehat{\sigma}^2} \label{eq:c3-exp3b}.
\end{align}
Combining \eqref{eq:c3-exp3a} and \eqref{eq:c3-exp3b}:
{\small
\begin{align}
&\mathbb{E}[ X_c \, | \, a, b, \text{Case 3} ] \nonumber \\
&= \rho \mathbb{E}[ X_p \, | \, a, b, \text{Case 3} ] +  \mathbb{E}[  \mathcal{N}(0, \sigma^2) \, | \, a, b, \text{Case 3} ] \nonumber \\
&=  \frac{\rho \sigma_s^2 [\rho  \,  \widehat{\sigma}^2 a +  (\sigma^2 + \lambda_t^2) b] + \sigma^2 [(\sigma_s^2 + \widehat{\sigma}^2)a - \rho \sigma_s^2 b ]}{\sigma_s^2(\sigma^2 + \lambda_t^2 + \widehat{\sigma}^2) + \lambda_t^2 \widehat{\sigma}^2} \nonumber \\
& = \frac{\sigma_s^2(\sigma^2 + \widehat{\sigma}^2) a + \rho \sigma_s^2  \lambda_t^2 b }{\sigma_s^2(\sigma^2 + \lambda_t^2 + \widehat{\sigma}^2) + \lambda_t^2 \widehat{\sigma}^2}. \label{eq:c3-exp}
\end{align}
}

Result \ref{result2} is obtained by combining the above calculations.
Considering \eqref{eq:etasum} and \eqref{eq:1},
\begin{equation} 
 \eta(a,b)\!=\!\frac{\sum_{j=3}^4 f(\text{Case j}, a, b)\mathbb{E}[ X_c | a, b, \text{Case j} ]}{\sum_{i = 1}^4 f(\text{Case i}, a, b)}, \label{eq:etasum_new}
\end{equation} 
which results in
the denoiser presented in \eqref{eq:denoiser} - \eqref{eq:S} with
$S_{(a, b)} = \sum_{i=1}^4 f(\text{Case i}, a, b)$,
where the probabilities are calculated in  \eqref{eq:cj-prob} and \eqref{eq:c3-prob}, $\epsilon_3 \, \mu^3_{(a, b)} = f(\text{Case 3}, a, b)$ and
$\epsilon_4 \, \mu^4_{(a, b)} = f(\text{Case 4}, a, b)$,
and finally with $\mathbb{E}[ X_c \, | \, a, b, \text{Case 3} ]$ and $\mathbb{E}[ X_c \, | \, a, b, \text{Case 4} ]$ calculated in \eqref{eq:c3-exp} and \eqref{eq:c4-exp}, respectively.

\subsection{State Evolution for BDD}
Using the results from the previous section, specifically the form of the denoiser in \eqref{eq:denoiser},
we can calculate the SE equations \eqref{eq:SE2}. Letting $\delta = \frac{M}{N}$, we have $\lambda_0^2 = \frac{1}{\delta}\mathbb{E}[X_c^2] + \sigma_z^2$ and for $t \geq 0$,
 \begin{align*}
\lambda_{t+1}^2 =  \sigma_z^2 + \frac{1}{\delta} \mathbb{E}\left[ (\eta_t(X_c+\lambda_t Z_{1}, X_p+\widehat{\sigma} Z_{2})-X_c)^2\right],
\end{align*}
where $\eta_t(\cdot, \cdot)$ 
is defined in \eqref{eq:denoiser},
and the RVs 
$Z_1$ and $Z_2$ are both zero mean unit norm Gaussian, and are independent of the
RVs 
$X_p$ and $X_c$, which are distributed according to the prior distributions of $x_p$ and $x_c$.  The expectation is with respect to $Z_1, Z_2, X_p$, and $X_c$, where $X_p$ and $X_c$ are dependent. Similarly to the SE equations for the BG signal model, it seems infeasible to find a closed-form value for the expectation in the SE equations due to the complicated form of the denoiser given in~\eqref{eq:denoiser}. We estimate these values numerically in the following section.

\section{Numerical Results} \label{sec:sims}

Here, we present the empirical performance of AMP-SI for the BG and BDD signal models. All numerical results were generated using MATLAB.

\begin{figure}[t!]
\centering
\vspace*{-3mm}
\hspace{-2mm}
\includegraphics[width=.4\textwidth]{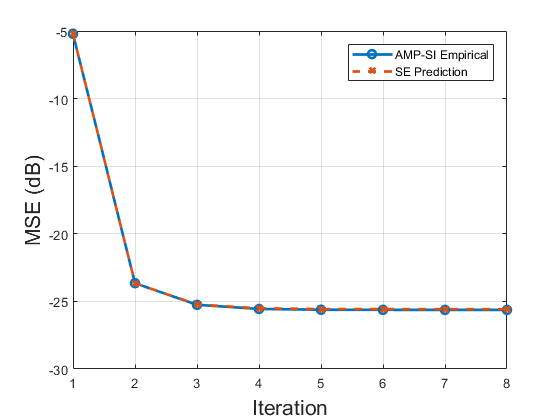}
\caption{Empirical performance of AMP-SI and performance predicted by SE
across iterations. (BG signal, $N=10000$, $M=3000$, $\sigma_z=0.1$,  $\epsilon=0.3$, $\widehat{\sigma}=0.1$.)}
\label{fig:BG_SE}
\vspace*{-4.5mm}
\end{figure}

{\bf BG signal:}\
Fig.~\ref{fig:BG_SE} presents the empirical performance of AMP-SI on a BG signal and the SE prediction of its performance. For this experiment, the signal has dimension $N=10000$,
the SI has standard deviation $\widehat{\sigma} = 0.10$, the number of measurements is $M = 3000$,
and the measurement noise standard deviation is $\sigma_z = 0.10$.
We set $\epsilon = 0.30$ so that approximately $30\%$ of the entries in the signal are nonzero.
The measurement matrix $A \in \mathbb{R}^{M \times N}$ has i.i.d.\ standard Gaussian entries.
The empirical normalized MSE
for AMP-SI is averaged over 20 trials of a BG recovery problem.
We are also plotting MSE results predicted by SE, and it can be seen that the SE prediction
accurately tracks the empiricial performance of AMP-SI.

{\bf BDD signal:}\
Fig.~\ref{fig:batch_mmse} presents experimental results for recovering a signal $x_c$ 
over 15 time batches following the BDD model of Section~\ref{sec:BDD} using AMP, AMP-SI, and DCS-AMP, another AMP-based algorithm for time-varying signals~\cite{DCSAMP}. 
In each time batch, the SI is
the pseudo-data output of AMP-SI in the previous batch, except for the Batch~1 where 
SI is unavailable and we resort to standard AMP. 
For DCS-AMP, we implement the algorithm in filtering mode to match our SI setting. 
The signal $x_c$ is of dimension $N=10000$, the steady-state standard deviation is
$\sigma_s = 1$, the correlation among nonzero entries is $\rho = 0.95$,
and the measurement noise has standard deviation $\sigma_z = 0.077$, which corresponds to  
$\text{SNR}=20\text{dB}$.
The empirical MSE is averaged over 100 trials. For each batch, AMP-based approaches 
often converge within 10--20 iterations, and 30 iterations are used to play it safe. 
We set $\epsilon_1 = 0.80$, $\epsilon_2 = \epsilon_4 = 0.01$, and $\epsilon_3 = 0.18$ so that there are approximately $K = N(\epsilon_3+\epsilon_4)=1900$
nonzero entries per signal. The measurement matrix has i.i.d.\ standard Gaussian entries in each batch, and the number of measurements is
$M=3000$. It can be seen that AMP-SI outperforms DCS-AMP in every batch (except Batch~1 where they resort to AMP). This further supports our belief in AMP-SI's Bayes-optimality properties.

Inspired by~\cite{DCSAMP}, we investigated how transition probabilities between different 
BDD states affected the performance of AMP-SI. 
Within the problem setting of Fig.~\ref{fig:batch_mmse},
we varied $\epsilon_2$, $\epsilon_4 \in \left[0.01, 0.11\right]$, 
and $\rho\in \left[0.05,0.999\right]$.
We found that 
AMP-SI out-performed AMP at  
Batch 15 for all 
configurations. The performance gap was largest for large $\rho$ and small
$\epsilon_2$ and $\epsilon_4$; the gap narrowed as $\rho$ decreased and
$\epsilon_2$ and $\epsilon_4$ increased. 
We also experimented with holding
 $\epsilon_2=\epsilon_4=0.05$ constant and varied $\epsilon_3$ and $\rho$. 
Again, as $\rho$ decreased, the gap between AMP-SI and AMP shrank. 
That said, for large $\eps_3$, AMP-SI still improved reconstruction quality for small
$\rho$. 

\begin{figure}[t!]
\centering
\vspace*{-2.5mm}
\hspace{-5mm}
\includegraphics[width=.5\textwidth]{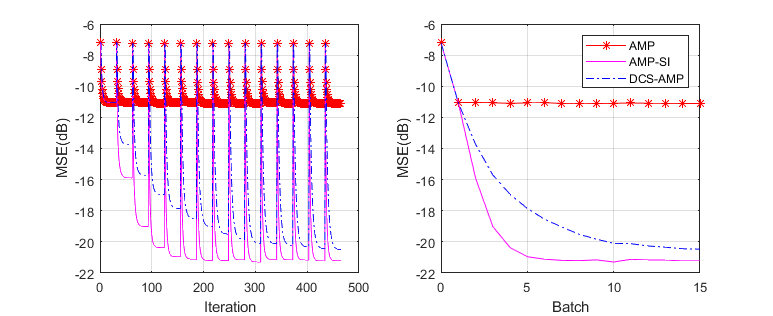}
\caption{Empirical performance of DCS-AMP in filter mode \cite{DCSAMP}, AMP, and AMP-SI 
as a function of 
(left) iterations and (right) batches, where the algorithms
each spend 30 iterations per batch.
(BDD signal, $N=10000$,  $M=3000$, $\sigma_s=1$, $\rho=0.95$, $\sigma_z=0.077$,
$\epsilon_1=0.80$, $\epsilon_2=\epsilon_4=0.01$, $\epsilon_3=0.18$. The tuning parameters for DCS-AMP are determined by converting BDD parameters to corresponding Gauss-Markov parameters \cite{DCSAMP}.)
}
\label{fig:batch_mmse}
\vspace*{-4.5mm}
\end{figure}

{\bf SE for BDD:}\
To highlight the advantages of SI, 
Fig.~\ref{fig:SE} shows the recovery quality predicted by SE. Here, the SI dependencies remain identical to the experiments used for Fig.~\ref{fig:batch_mmse} as we vary the number of measurements $M$ (to show different $\delta=M/N$) and
$\epsilon_1$ and $\epsilon_3$ (to show different
percentages of nonzeros, $\gamma=K/N$). We also hold the measurement noise $\sigma_z=0.01$ constant.
To vary $\gamma$, we keep $\epsilon_2=\epsilon_4 = 0.01$ while modifying the probability
of the drift case, $\epsilon_3$, accordingly.
In each panel, the horizontal axis corresponds to $\delta$,
the vertical axis to $\gamma$,
and shades of gray to the SE prediction of the MSE.
Batch~1 corresponds to the first time the signal is recovered without SI,
Batch~3 uses recovered signals from the second batch as SI, and Batch~10 uses the recovered signal from Batch~9 as SI.
The high-quality dark gray
region in the upper right portion of each panel is expanding,
while the low-quality light gray region is shrinking, showing improved signal recovery
due to the SI.
More specifically, the proportion of the high-quality dark grey region in each subplot is 0.413 (Batch~1), 0.448 (Batch~3), and 0.562 (Batch~10).
It can be seen that the same MSE quality is obtained from a measurement rate $\delta$
lower than without SI.

 \begin{figure*}[th!]
\centering
\hspace*{2mm}\includegraphics[scale=.32]{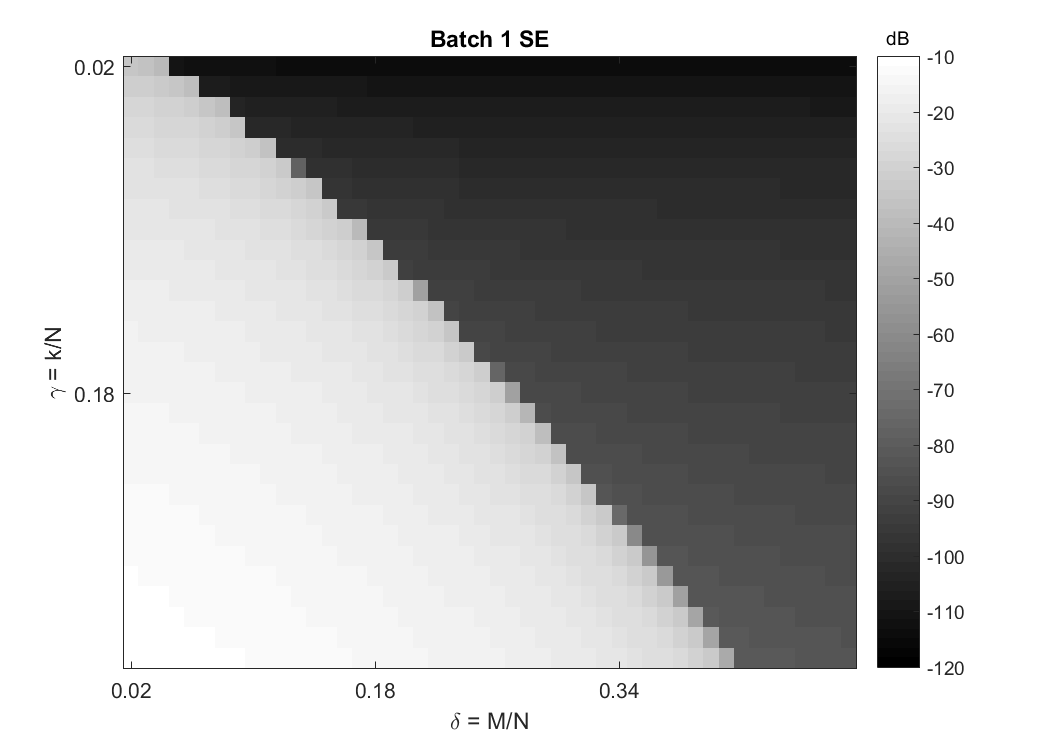}
\hspace*{-11mm}\includegraphics[scale=.32]{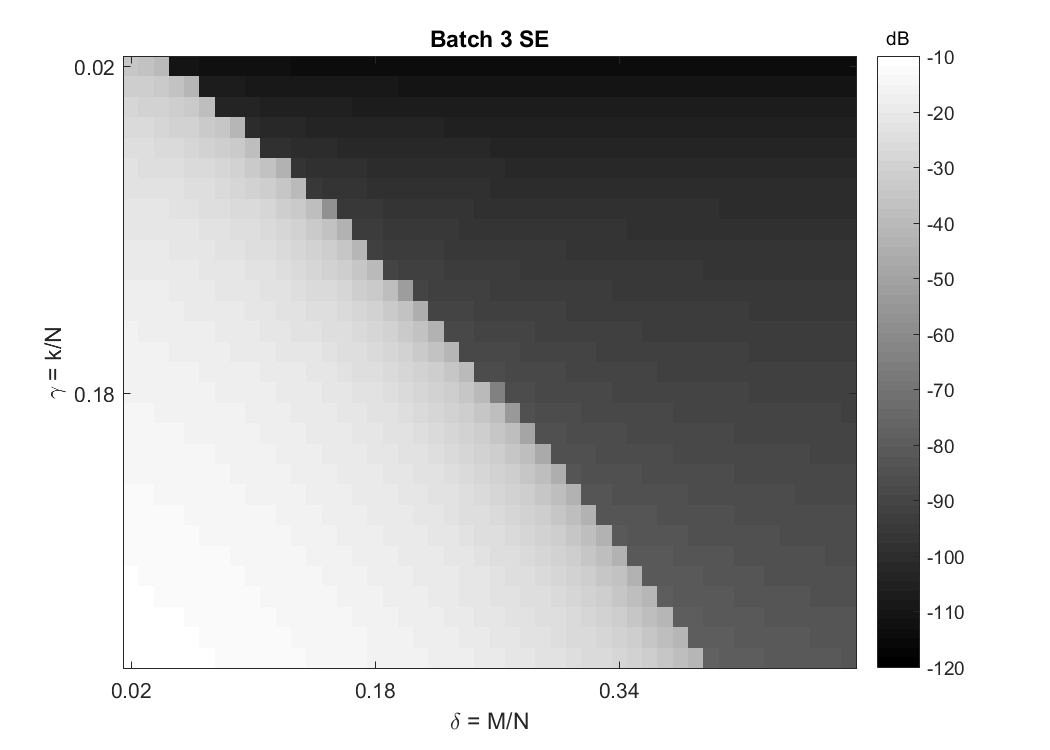}
\hspace*{-11mm}\includegraphics[scale=.32]{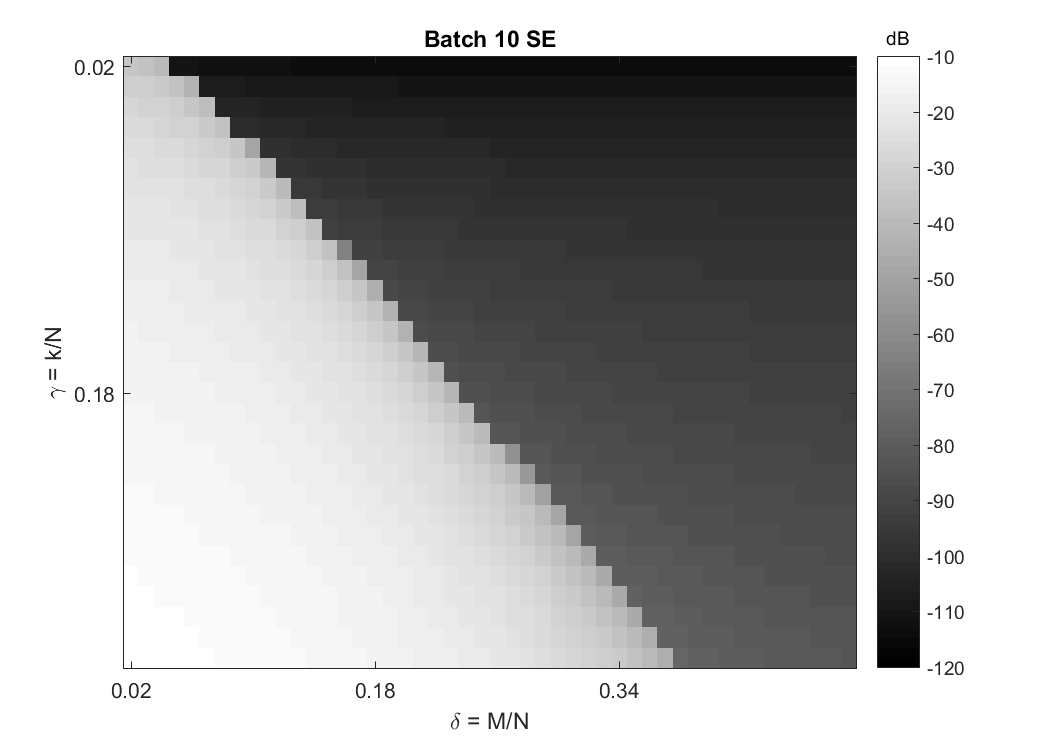}
\caption{AMP-SI for BDD signals.
The MSEs predicted by SE are plotted as shades of gray; they are functions of measurement rate $\delta = \frac{M}{N}$ and sparsity rate $\gamma = \frac{k}{N}$.
From left to right: Batch~1 (without SI),
Batch~3 (SI=Batch~2), and Batch~10 (SI=Batch~9). The `good' dark gray region (upper right corners) expands with more SI.
}
\label{fig:SE}
\vspace*{-4.5mm}
\end{figure*}

{\bf Channel estimation with Toeplitz matrices:}\
So far we used i.i.d.\ Gaussian matrices, and we now transition to Toeplitz matrices
in order to demonstrate that AMP-SI is suitable for channel estimation
(details in Section~\ref{sec:BDD}). Based on \eqref{eq:Toeplitz}, the channel estimation problem deviates from the BDD model in three aspects.
First, as mentioned,
$A$ is Toeplitz rather than i.i.d.\ Gaussian.
It is well known that for non-i.i.d.\ sensing matrices, the standard AMP prescribed by (\ref{eq:AMP1}) and (\ref{eq:AMP2}) often suffers from divergence over iterations.
A common approach to improve convergence of iterative algorithms is damping; in AMP, the standard iteration (\ref{eq:AMP2}) is replaced by $x^{t+1}=\lambda x^t + (1-\lambda) \eta_t( x^t + A^Tr^t)$.
Rangan et al.\ \cite{Rangan2014} demonstrate that damping is effective in aiding the convergence of AMP for some non-i.i.d.\ sensing matrices. It should be noted that supporting theory for AMP-based algorithms is only rigorous for certain classes of random 
matrices~\cite{BayMont11, RanganVAMP}, which exclude Toeplitz 
matrices. Thus 
we evaluate the performance of AMP-SI for Toeplitz matrices by comparing empirical reconstruction results to  standard AMP-SI settings. Lastly, for a pilot sequence $p$, the number of rows of the measurement matrix, $M$, equals $\text{length}(p)+N-1$, which typically exceeds $N$, the number of columns.
This inverse problem is expansive ($M>N$) instead of compressive ($M<N)$, where we remind the reader that AMP and SE theory support arbitrary $\delta>0$
where $\delta = \frac{M}{N}$.

Our experiment had 5 time batches. We set the length of the channel response $N$ to $4000$,
the length of the pilot sequence $\text{length}(p)=1001$, the standard deviation of the steady signal $\sigma_s=1$, the decay rate of nonzeros $\rho=0.95$, and the measurement noise standard deviation $\sigma_z \in \{0.01, 0.1, 1\}$. This setting corresponds to SNR= 0dB, 20dB and 40dB, and $\delta=1.25$. For BDD model parameters, we set $\epsilon_1 = 0.78$, $\epsilon_2 = \epsilon_4 = 0.01$. Thus at each time batch, $21\%$ of the entries of the channel response are nonzero. The individual entries of the pilot $p$ are $\pm 1/\sqrt{\text{length}(p)}=\pm 0.0316$, each with probability 0.5.
We performed damping using parameter  $\lambda=0.9$.
Table~\ref{tab:Toeptab} demonstrates the empirical channel estimation performance of AMP-SI averaged over 50 realizations. Compared to standard AMP (batch~1 in Table~\ref{tab:Toeptab}), AMP-SI consistently achieves lower MSE levels starting from batch~2. One striking observation from Fig.~\ref{fig:Toeplitz} is the similar performance of AMP-SI for Toeplitz (channel estimation) and i.i.d.\ matrices.
This similarity leads us to conjecture that for the given BDD signal model, SE prediction tracks the performance of AMP/AMP-SI with Toeplitz matrices as well as the i.i.d.\ Gaussian case. The conjecture is further evident from Table~\ref{tab:Toepbth5}. Observations from other BDD time batches resemble
batch~5 (Table~\ref{tab:Toepbth5}) and are not included.

\begin{figure}[ht!]
\centering
\vspace*{-2.5mm}
\hspace{-2mm}
\includegraphics[width=.5\textwidth]{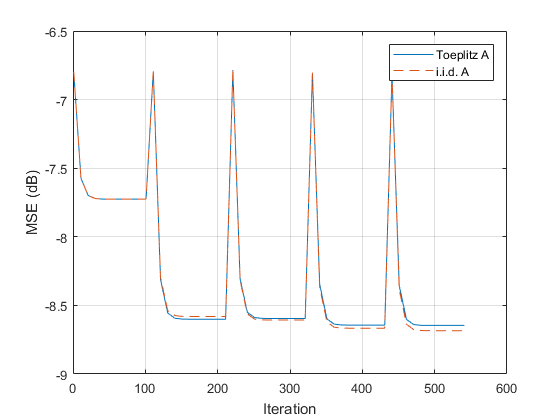}
\caption{Empirical AMP-SI performance with i.i.d.\ and Toeplitz sensing matrices.
(BDD signal, SNR=0dB, averaged over 100 realizations.)}
\label{fig:Toeplitz}
\vspace*{-4.5mm}
\end{figure}

\begin{table}[]
\centering
\begin{tabular}{|c|c|c|c|}
\hline
 & \multicolumn{3}{c|}{\textbf{Channel Estimation MSE(dB)}} \\ \hline
\textbf{SNR} & \textbf{Batch 1} & \textbf{\begin{tabular}[c]{@{}c@{}}Batch 2 \end{tabular}} & \textbf{Batch 5} \\ \hline
0dB & -7.71 & -8.61 & -8.68 \\ \hline
20dB &  -23.49 & -24.74 & -24.86 \\ \hline
40dB & -45.41 & -45.83 & -45.82 \\ \hline
\end{tabular}
\caption{Empirical AMP-SI performance for channel estimation.
(BDD signal, averaged over 50 realizations.)
} \label{tab:Toeptab}
\end{table}

\begin{table}[]
\centering
\begin{tabular}{|c|c|c|c|}
\hline
 & \multicolumn{3}{c|}{\textbf{AMP-SI Performance in MSE(dB) at Time Batch 5}} \\ \hline
\textbf{SNR} & \textbf{~~i.i.d.\ A~~} & \textbf{~~SE Prediction ~~} & \textbf{Toeplitz A} \\ \hline
0dB & -8.68 & -8.64 & -8.65 \\ \hline
20dB & -24.86 & -24.86 & -24.77 \\ \hline
40dB & -45.82 & -45.91 & -45.86 \\ \hline
\end{tabular}
\caption{Empirical AMP-SI performance for i.i.d. matrices,
SE predictions, and empirical performance for Toeplitz matrices.
(BDD signal, Batch~5, averaged over 50 realizations.)
} \label{tab:Toepbth5}
\vspace*{-4.5mm}
\end{table}

\section{Challenges and Future Work}  \label{sec:futurework}

In this work, we presented AMP-SI, a suite of approximate message passing (AMP) based algorithms that utilize \emph{side information} (SI) to aid in signal recovery using \emph{conditional} denoisers.
We derive conditional denoisers for a Bernoulli-Gaussian (BG) signal model and a more complicated time-varying birth-death-drift (BDD) signal model, motivated by channel estimation.
We also conjectured state evolution (SE) properties.
Numerical experiments show that the proposed SE accurately tracks the performance of AMP-SI,
and that AMP-SI achieves the same MSE as AMP using a lower measurement rate.

To simulate the channel estimation task, we additionally consider a Toeplitz measurement matrix as opposed to the standard Gaussian i.i.d.\ matrix. Our results show that AMP-SI is able to obtain a lower MSE than AMP for such a setting.
A challenge and future direction with this line of work is that the current theoretical guarantees for AMP assume that $A$ is an i.i.d.\ matrix.
Although AMP often
diverges when non-i.i.d.\ matrices are used, there is empirical evidence that AMP can successfully perform deconvolution and utilize other structures in various settings \cite{kamilov2013sparse,barbier2015approximate}. 
We believe our AMP-SI framework will lead to new applications in a broad class of problems while also presenting interesting theoretical challenges.

\section*{Acknowledgments}
We are grateful to Yavuz Yapici and Ismail Guvenc who helped us formulate the BDD model, and to Chethan Anjinappa whose numerical ray tracing simulation (Fig.~\ref{fig:raytrace}) helped confirm the model.
Our work originated from earlier work on AMP with SI,
which was joint with Tina Woolf.
We thank Hangjin Liu for conversations about using the recent results by Berthier et al. \cite{berthier2017} for our SE proofs.
Finally, we thank Junan Zhu and Yanting Ma for helping us formulate the problem
and master some of the deeper technical details.

\appendix
\subsection{Proof of Lemma \ref{lem:jointGauss}}   \label{app:jointGauss}

Recall from the Lemma statement that $A = \rho X + \mathcal{N}(0, \sigma_a^2)$ and $B = X + \mathcal{N}(0, \sigma_b^2)$ where $X \sim \mathcal{N}(0, \sigma_x^2)$.

Then from Bayes' rule, $f(a,b) = f(b)f(a \, | \, b)$ and computing $f(a \, | \, b)$ we have:
\begin{align*}
f(a \, | \, b) &= \int_x f(a,x \, | \, b) \,dx = \int_x f(x \, | \, b) f(a  \, | \, b,x) \,dx \\
&\stackrel{(1)}{=} \int_x \frac{f(x)f(b \, | \, x)}{f(b)} \psi_{\sigma_a^2}(a-\rho x)\, dx \\
&= \int_x \frac{ \psi_{\sigma_x^2}(x) \psi_{\sigma_b^2}(b-x)}{f(b)} \psi_{\sigma_a^2}(a-\rho x) \,dx,
\end{align*}
where equality (1) relies on Bayes' rule applied to $f(x \, | \, b)$.
Therefore,
\begin{align*}
f(a,b) &=  f(b)f(a \, | \, b)  \\
&= \int_x \psi_{\sigma_x^2}(x) \psi_{\sigma_b^2}(b-x) \psi_{\sigma_a^2}(a-\rho x) \,dx  \\
&\stackrel{(2)}{=} \frac{1}{\rho} \psi_{\sigma_x^2 + \sigma_b^2}(b) \psi_{\frac{\sigma_x^2 \sigma_b^2}{\sigma_x^2 + \sigma_b^2} + \frac{\sigma_a^2}{\rho^2}} \left( \frac{\sigma_x^2}{\sigma_x^2+\sigma_b^2} b - \frac{1}{\rho}a \right),
\end{align*}
where equality (2) uses Lemma \ref{lem:gauss_prod}.

\subsection{Proof of Lemma \ref{lem:jointCond}} \label{app:jointCond}

Recall from the Lemma statement that $A = \rho X + \mathcal{N}(0, \sigma_a^2)$ and $B = X + \mathcal{N}(0, \sigma_b^2)$ where $X \sim \mathcal{N}(0, \sigma_x^2)$.

Because $X$, $A$, and $B$ are jointly Gaussian RVs, the MMSE-optimal estimator for $X$ conditioned on $a$ and $b$ is linear,
\begin{equation}
\widehat{x}=\mathbb{E}[X \, |\, a, b] = \alpha a + \beta b + \gamma,
\label{eq:linear_MMSE_estimator}
\end{equation}
where $\alpha$, $\beta$, and $\gamma$ are constants.
A well known result (see, e.g., Theorem 9.1 of \cite{MMSEwiki})
states that
$\widehat{x} = W \begin{bmatrix}a\\ b\end{bmatrix} + U$,
where
\begin{align*}
W=C_1^T (C_2)^{-1}, &\quad
U=\mathbb{E}[X]-W \mathbb{E}\begin{bmatrix}A\\ B \end{bmatrix}, \\
C_1=\text{Cov}\left(X,\begin{bmatrix}A \\ B  \end{bmatrix}\right), & \quad C_2=\text{Cov}\left(\begin{bmatrix}A \\ B \end{bmatrix},\begin{bmatrix} A \\ B \end{bmatrix}\right).
\end{align*}
We compute these terms one by one.
First, $X$, $A$, and $B$ all have zero mean, and so $U=0$,
which implies that the constant $\gamma$ in the linear form \eqref{eq:linear_MMSE_estimator}
is zero. Second,
\[
C_1
=\text{Cov}\left(X,\begin{bmatrix}A \\ B \end{bmatrix}\right)
= \begin{bmatrix}\mathbb{E}[XA]\\ \mathbb{E}[X B]\end{bmatrix},
\]
because the zero means ensure that only the cross terms
$\mathbb{E}[X A]$ and $\mathbb{E}[X B]$ appear in the expression for $C_1$.
The cross terms are computed as
\begin{align*}
\mathbb{E}[X A] &=\mathbb{E}[X (\rho X+\mathcal{N}(0, \sigma_a^2))] =\rho \sigma^2_x, \\
\mathbb{E}[X B] &= \mathbb{E}[ X (X+\mathcal{N}(0, \sigma_b^2))] =\sigma_x^2.
\end{align*}
Therefore, $C_1 = \sigma^2_x \begin{bmatrix}\rho\\ 1 \end{bmatrix}$.
Third,
\[
C_2
=\text{Cov}\left(\begin{bmatrix}A \\ B\end{bmatrix},\begin{bmatrix}A \\ B\end{bmatrix}\right)
=\mathbb{E}\left[ \begin{bmatrix}A \\ B\end{bmatrix}[A  \,\, B] \right], \]
where once again only the cross terms need be computed.
These cross terms are
({\em i}) $\mathbb{E}[A^2]=\rho^2 \sigma^2_x+\sigma_a^2$;
({\em ii}) $\mathbb{E}[B^2]=\sigma^2_x+\sigma_b^2$;
and
({\em iii})
\begin{align*}
\mathbb{E}[AB]=\mathbb{E}[BA] &=\mathbb{E}[(\rho X + \mathcal{N}(0, \sigma_a^2))( X + \mathcal{N}(0, \sigma_b^2))] \\
&=\rho\sigma^2_x.
\end{align*}
The MMSE-optimal estimator is
\begin{equation}
\begin{split}
&\mathbb{E}[X\, | \,  a, b]
= W \begin{bmatrix} a \\ b\end{bmatrix}
= C_1^T (C_2)^{-1} \begin{bmatrix}a \\ b \end{bmatrix} \\
&= [\rho\sigma^2_x\ \sigma^2_x]
\left[ {\begin{array}{cc}
\rho^2 \sigma^2_x+\sigma_a^2 & \rho\sigma^2_x \\
\rho\sigma^2_x & \sigma^2_x+\sigma_b^2
\end{array} } \right]^{-1}
\begin{bmatrix}a \\ b\end{bmatrix} \\
&= \frac{\rho \sigma_x^2 \,  \sigma_b^2 a +  \sigma_x^2 \sigma_a^2 b}{\sigma_x^2 \, (\sigma_a^2 + \rho^2 \sigma_b^2) + \sigma_a^2 \sigma_b^2}.
\label{eq:c3-exp}
\end{split}
\end{equation}

\subsection{Fixed points of AMP-SI SE with Gaussian SI} \label{app:asilomar}

This appendix will show that when the SI is a Gaussian-noise corrupted
observation of the true signal,
i.e., $\widetilde{X} = X + \mathcal{N}(0,\widehat{\sigma}^2)$,
the fixed points of AMP-SI SE~\eqref{eq:SE2} coincide with the fixed points of AMP SE~\eqref{eq:AMP_SE} with `effective' measurement rate
$\delta_{eff} = \delta/\mu$ and
`effective' measurement noise variance $\sigma_{eff}^2 = \mu \sigma_{z}^2 $ where $0 \leq \mu \leq 1$ and
$\mu$ depends on the pdf of the signal and the SI noise variance $\widehat{\sigma}^2$.

Before demonstrating the aforementioned Bayes-optimality property of AMP-SI, we use matched filter arguments to provide a simplified representation of the conditional denoiser of \eqref{eq:eta_CAMP} when the SI is the signal viewed with AWGN.
In calculating the AMP-SI denoiser \eqref{eq:eta_CAMP}, we want to calculate the expectation of $X$ conditioned
on the pseudo data, $X+\lambda_t Z_1=a$, and SI, $X+ \widehat{\sigma}Z_2=b$, where $Z_1$ and $Z_2$ are
independent, standard Gaussian RVs. We define signal and noise vectors as $s=[1\ 1]^T$  and
$v=[\lambda_t Z_1\  \widehat{\sigma} Z_2]^T$, respectively, where $[\cdot]^T$ is the transpose operator.
The matched filter estimates the unknown $X$ by computing the inner product between
\[
\left[ {\begin{array}{c}a \\ b\end{array}} \right]
=
\left[ {\begin{array}{c}X+\lambda_t Z_1 \\ X+\widehat{\sigma} Z_2 \end{array}} \right]
=sX+v,
\]
and a matched filter $h \in \mathbb{R}^2$. An optimal $h^*$ that maximizes
the signal to noise ratio while having unit norm is computed by inverting
$R_v=E[vv^T]$, the auto-covariance matrix of $v$, 
\[
h^*= { (R_v)^{-1} s} /{\|(R_v)^{-1} s\|}.
\]
It can be shown that $h^*=[ \widehat{\sigma}^2\ \lambda^2_t]^T/(\widehat{\sigma}^2+\lambda_t^2)$,
and the inner product is defined as $\mu^t(a,b):$
\begin{equation}
\mu^t(a,b) = \langle [a\ b]^T, h^* \rangle
= \frac{a \widehat{\sigma}^2+b\lambda_t^2}{\widehat{\sigma}^2+\lambda_t^2}.
\label{eq:mu_def}
\end{equation}
Note that $\mu^t(X+\lambda_t Z_1, X+ \widehat{\sigma} Z_2)$ equals
\begin{equation*}
\frac{(X+\lambda_t Z_1) \widehat{\sigma}^2+(X+ \widehat{\sigma}Z_2)\lambda_t^2}{\widehat{\sigma}^2+\lambda_t^2} \overset{d}{=}
X+\sigma_t Z,
\end{equation*}
where $Z$ is standard Gaussian, $\overset{d}{=}$ denotes equality in distribution,
and the variance term, $(\sigma_t)^2$, is
\begin{equation}
(\sigma_t)^2 =
\frac{(\lambda_t \widehat{\sigma}^2)^2+(\widehat{\sigma} \lambda_t^2)^2}{( \widehat{\sigma}^2+\lambda_t^2)^2}
=
\frac{\lambda_t^2\widehat{\sigma}^2}{\widehat{\sigma}^2+\lambda_t^2}.
\label{eq:sigma_def}
\end{equation}

The above provides us with the following simplification of the AMP-SI denoiser \eqref{eq:eta_CAMP} for SI with AWGN,
\begin{equation}
\begin{split}
& \eta_t(a, b) = \mathbb{E}[X | X + \sigma^t Z =  \mu^t(a, b)],
\label{eq:rep2}
\end{split}
\end{equation}
where $\mu^t(a,b)$ and $\sigma^t$ are defined in \eqref{eq:mu_def} and \eqref{eq:sigma_def}.
We note that $\mu^t$ is a function of $(a, b)$, but for brevity we drop this dependence in the following.
Considering~\eqref{eq:SE2} and~\eqref{eq:rep2}, 
\begin{equation}
\eta_{t}(X + \lambda_{t}Z_1, X + \widehat{\sigma} Z_2) =  \mathbb{E}[X | X + \sigma^{t} Z].
\label{eq:eta_t_rep}
\end{equation}
We simplify the SE equations \eqref{eq:SE2} using \eqref{eq:eta_t_rep} and the definition of $\sigma^t$ in \eqref{eq:sigma_def}.  Let $\lambda_0 = \sigma_z^2 + \mathbb{E}[X^2]/\delta$ and for $t \geq 0$,
{\small
\begin{equation}
\hspace*{-.15mm}
\lambda_t^2 = \sigma_z^2 + \frac{1}{\delta}\mathbb{E}\!\left[\!\left(\mathbb{E}\!\left[X\!\left| X
\!+\!
\sqrt{\frac{\lambda_{t-1}^2 \widehat{\sigma}^2}{ \widehat{\sigma}^2 + \lambda_{t-1}^2}} \, Z\right]\right .\!\!-\! X\right)^{\!2}\right]\!.\label{eq:SE3}
\end{equation}
}
\hspace{-4pt}The results in \eqref{eq:rep2} and \eqref{eq:SE3} provide
a simplified way to calculate the conditional denoiser of \eqref{eq:eta_CAMP}
and the SE
{\em when the signal and the SI are related through Gaussian noise}.
Moreover, at the stationary point of \eqref{eq:SE3} we have
{\small
\begin{equation}
\lambda^2 = \sigma_z^2 +  \frac{1}{\delta}\mathbb{E}\! \left[ \!\left( \! \mathbb{E}\!\left[X \left|
X \!+\! \sqrt{\frac{\lambda^2 \widehat{\sigma}^2}{ \widehat{\sigma}^2  \! + \!  \lambda^2}} \, Z\right]\right. \!-\! X \! \right)^{\!2}\right]\!,
\label{eq:SE4}
\end{equation}
}
\hspace{-4pt}where $\lambda^2$ is the scalar channel variance. Comparing~\eqref{eq:AMP_SE} (SE without SI) and \eqref{eq:SE4}, we denote the variance in the conditional expectation
by $\widetilde{\lambda}^2 = \frac{\lambda^2\widehat{\sigma}^2}{\widehat{\sigma}^2 + \lambda^2}$.
Note that $\lambda^2 = \frac{\widetilde{\lambda}^2\widehat{\sigma}^2}{ \widehat{\sigma}^2 - \widetilde{\lambda}^2} \geq 0$,
because
$ \widetilde{\lambda}^2 \leq \widehat{\sigma}^2$, and we can rewrite the above as
\begin{equation}
\widetilde{\lambda}^2  \!=\!
\frac{ (\widehat{\sigma}^2 - \widetilde{\lambda}^2)\sigma_z^2}{ \widehat{\sigma}^2} +
\frac{1}{\frac{\delta  \widehat{\sigma}^2}{\sigma^2 - \widetilde{\lambda}^2}}
\mathbb{E}\!\left[\!\left(\mathbb{E}[X | X + \widetilde{\lambda}  Z] \!-\! X\right)^2\right].
\label{eq:SE5}
\end{equation}

We see that
AMP-SI SE~\eqref{eq:SE2} has fixed points coinciding with the fixed points of
standard AMP SE~\eqref{eq:AMP_SE} with `effective' measurement rate
$\delta_{eff} = \delta\left(\frac{\widehat{\sigma}^2 + \lambda^2}{\widehat{\sigma}^2}\right)$ and
`effective' measurement noise variance $\sigma_{eff}^2 = \left(\frac{\widehat{\sigma}^2}{\widehat{\sigma}^2 + \lambda^2}\right)\sigma_z^2$ where $\widehat{\sigma}^2$ is the noise in the SI and $\lambda^2$ is the stationary point of \eqref{eq:SE3}.
This effective change in $\delta$ and $\sigma^2$
implies that the incorporation of SI with AWGN via the AMP-SI algorithm gives us signal recovery for a standard
(without SI)
linear regression problem \eqref{eq:hdreg} with \emph{more} measurements and/or \emph{reduced} measurement noise variance than our own,
and the effect becomes more pronounced, as the noise variance in the SI, $\widehat{\sigma}^2$, gets small.

The above analysis relies on the fact that for the conditional expectation denoiser in standard
(without SI)
AMP~\eqref{eq:AMP1}-\eqref{eq:AMP2}, the corresponding SE
equation~\eqref{eq:AMP_SE} in its convergent states coincides with Tanaka's fixed point equation~\cite{Tanaka2002}, 
ensuring that if AMP runs until it converges, the result provides the best possible MSE achieved by any algorithm under certain conditions.
(These conditions on $\delta$ and $\epsilon$, while outside the scope of this paper,
ensure that there is a single solution to Tanaka's fixed point equation, since multiple solutions
may create a disparity between the MSE of AMP and the MMSE~\cite{Krzakala2012probabilistic},
implying that AMP-SI might be sub-optimal in such cases.)
Our analysis relies heavily on the Gaussianity of the SI noise and
may not easily be generalized.

\subsection{Theorem \ref{thm:SE} Example}   \label{app:SEtheory}

As an example, we study the simple Gaussian-Gaussian (GG) case.   In the GG model one wants to recover a signal having i.i.d.\ zero-mean Gaussian elements with variance $\sigma_X^2$ with SI equal to the signal plus additive white Gaussian noise (AWGN) with variance $\widehat{\sigma}^2$, meaning $\sigma_{\widetilde{X}}^2 = \sigma_{X}^2 + \widehat{\sigma}^2$.  We will show that for assumptions \textbf{(A4)} and \textbf{(A5)} to be true, we need finite fourth moments $\sigma_X^4$ and $\widehat{\sigma}^4$.

In this case, using \eqref{eq:rep2} from Appendix~\ref{app:asilomar}, the denoiser $\eta_t: \mathbb{R}^2 \rightarrow \mathbb{R}$ is given by

\begin{equation}
\begin{split}
\eta_t(a,b)
&= \frac{ \sigma_X^2 (\sigma_{\widetilde{X}}^2 a +   \lambda_t^2 b )}{\sigma_X^2 ( \sigma_{\widetilde{X}}^2 + \lambda_t^2) +\lambda_t^2 \sigma_{\widetilde{X}}^2 }.
\label{eq:eta_2}
\end{split}
\end{equation}

Now we would like to prove the following assumptions needed for Theorem \ref{thm:SE} to hold:
\textbf{(A4)} For $t\geq 0$, the denoisers $\eta_t: \mathbb{R}^{2} \rightarrow \mathbb{R}$ defined in \eqref{eq:eta_CAMP} are Lipschitz in their first argument.
\textbf{(A5)} For  any $2 \times 2$ covariance matrix $\Sigma$, let $(Z_1, Z'_1) \sim \mathcal{N}(0,\Sigma)$  independent of $(X, \widetilde{X}) \sim f(X, \widetilde{X}).$  Then for any $s, t \geq 0$,
\begin{equation}
\mathbb{E} [X \eta_t(X + Z_1, \widetilde{X})] < \infty,
\label{eq:finiteresult1}
\end{equation}
and
\begin{equation}
\mathbb{E} [\eta_t(X + Z_1, \widetilde{X}) \eta_s(X + Z'_1, \widetilde{X}) ] < \infty.
\label{eq:finiteresult2}
\end{equation}
 
Assumption $\textbf{(A4)}$ is straightforward using \eqref{eq:eta_2}: for fixed SI $ \widetilde{\mathbf{X}}$,
\eqref{eq:eta_2} suggests that for finite $\sigma_X^2$ and $\sigma_{\widetilde{X}}^2$,

\begin{equation*}
\lvert \eta_t(x, b) - \eta_t(y, b)\lvert  \leq \lvert x - y\lvert.
\end{equation*}

Next we consider assumption $\textbf{(A5)}$.  We will show~\eqref{eq:finiteresult1}
and then demonstrating~\eqref{eq:finiteresult2}
follows similarly.  First note

\begin{equation*}
\begin{split}
\mathbb{E}_{Z_1} [ \eta_t(X + Z_1, \widetilde{X})] &=\mathbb{E}_{Z_1} \Big[\frac{ \sigma_X^2 (\sigma_{\widetilde{X}}^2 (X + Z_1) +   \lambda_t^2  \widetilde{X} )}{\sigma_X^2 ( \sigma_{\widetilde{X}}^2 + \lambda_t^2) +\lambda_t^2 \sigma_{\widetilde{X}}^2 } \Big] \\
&= \frac{ \sigma_X^2 (\sigma_{\widetilde{X}}^2 X  +   \lambda_t^2  \widetilde{X} )}{\sigma_X^2 ( \sigma_{\widetilde{X}}^2 + \lambda_t^2) +\lambda_t^2 \sigma_{\widetilde{X}}^2 }.
\end{split}
\end{equation*}

Then using $\mathbb{E}_{X, \widetilde{X}}[X \widetilde{X}] = \mathbb{E}_{X,Z_{2}}[X( X + Z_{2})] = \mathbb{E}_X[X^2]  = \sigma_X^2$, we see
\begin{equation*}
\begin{split}
\mathbb{E}_{X, \widetilde{X}}&[X \, \mathbb{E}_{Z_1} [ \eta_t(X + Z_1, \widetilde{X})]] \\
&= \mathbb{E}_{X, \widetilde{X}} \Big[\frac{X \sigma_X^2 (\sigma_{\widetilde{X}}^2 X +   \lambda_t^2  \widetilde{X} )}{\sigma_X^2 ( \sigma_{\widetilde{X}}^2 + \lambda_t^2) +\lambda_t^2 \sigma_{\widetilde{X}}^2 } \Big] \\
& = \frac{\sigma_X^4 (\sigma_{\widetilde{X}}^2 +   \lambda_t^2  )}{\sigma_X^2 ( \sigma_{\widetilde{X}}^2 + \lambda_t^2) +\lambda_t^2 \sigma_{\widetilde{X}}^2 }.
\end{split}
\end{equation*}

For the above to be finite, we need finite $\sigma_X^4, \lambda_t^2,$ and $\widehat{\sigma}^2$.  For~\eqref{eq:finiteresult2} to hold, we need that $\widehat{\sigma}^4$ is finite.
We have shown that it is easy to demonstrate that the assumptions needed for Theorem \ref{thm:SE} hold in the GG case.

\bibliographystyle{abbrv}
\bibliography{bib}

\begin{thebibliography}{10}

\bibitem{MMSEwiki}
Minimum mean square error.
\newblock
  \url{https://en.wikipedia.org/wiki/Minimum_mean_square_error#cite_note-1},
  Retrieved July 7, 2017.

\bibitem{Arguello2011}
H.~Arguello and G.~Arce.
\newblock Code aperture optimization for spectrally agile compressive imaging.
\newblock {\em J. Opt. Soc. Am.}, 28(11):2400--2413, Nov. 2011.

\bibitem{barbier2017mutual}
J.~Barbier, N.~Macris, M.~Dia, and F.~Krzakala.
\newblock Mutual information and optimality of approximate message-passing in
  random linear estimation.
\newblock {\em arXiv preprint arXiv:1701.05823}, 2017.

\bibitem{barbier2015approximate}
J.~Barbier, C.~Sch{\"u}lke, and F.~Krzakala.
\newblock Approximate message-passing with spatially coupled structured
  operators, with applications to compressed sensing and sparse superposition
  codes.
\newblock {\em J. Stat. Mech-Theory E.}, 2015(5):P05013, May 2015.

\bibitem{BNMRW2017}
D.~Baron, A.~Ma, D.~Needell, C.~Rush, and T.~Woolf.
\newblock Conditional approximate message passing with side information.
\newblock In {\em Proc. Asilomar Conf. Signals, Systems, and Computers},
  Pacific Grove, CA, Nov. 2017.

\bibitem{BayMont11}
M.~Bayati and A.~Montanari.
\newblock The dynamics of message passing on dense graphs, with applications to
  compressed sensing.
\newblock {\em IEEE Trans. Inf. Theory}, 57(2):764--785, Feb. 2011.

\bibitem{berthier2017}
R.~Berthier, A.~Montanari, and P.-M. Nguyen.
\newblock State evolution for approximate message passing with non-separable
  functions.
\newblock {\em arXiv preprint arXiv:1708.03950}, 2017.

\bibitem{Boyd2011}
S.~Boyd, N.~Parikh, E.~Chu, B.~Peleato, and J.~Eckstein.
\newblock Distributed optimization and statistical learning via the alternating
  direction method of multipliers.
\newblock {\em Found. Trends Mach. Learn.}, 3(1):1--122, Jan. 2011.

\bibitem{Caire2004}
G.~Caire, R.~Muller, and T.~Tanaka.
\newblock Iterative multiuser joint decoding: Optimal power allocation and
  low-complexity implementation.
\newblock {\em IEEE Trans. Inf. Theory}, 50(9):1950--1973, Sept. 2004.

\bibitem{Candes2009}
E.~Cand\`es and B.~Recht.
\newblock Exact matrix completion via convex optimization.
\newblock {\em Found. Comput. Math.}, 9:717--772, Dec. 2009.

\bibitem{CandesUES}
E.~Cand\`{e}s and T.~Tao.
\newblock {Near-optimal signal recovery from random projections: Universal
  encoding strategies?}
\newblock {\em IEEE Trans. Inf. Theory}, 52(12):5406--5425, Dec. 2006.

\bibitem{Chen2008}
G.-H. Chen, J.~Tang, and S.~Leng.
\newblock Prior image constrained compressed sensing ({PICCS}): a method to
  accurately reconstruct dynamic {CT} images from highly undersampled
  projection data sets.
\newblock {\em Medical Physics}, 35(2):600--663, Feb. 2008.

\bibitem{Chen2016}
M.~Chen, F.~Renna, and M.~Rodrigues.
\newblock On the design of linear projections for compressive sensing with side
  information.
\newblock In {\em Proc. IEEE Int. Symp. Inf. Theory (ISIT)}, pages 670--674,
  Barcelona, Spain, July 2016.

\bibitem{DonohoBP}
S.~Chen, D.~Donoho, and M.~Saunders.
\newblock Atomic decomposition by basis pursuit.
\newblock {\em SIAM J. Sci. Comp.}, 20(1):33--61, Aug. 1998.

\bibitem{Cover06}
T.~Cover and J.~Thomas.
\newblock {\em Elements of Information Theory}.
\newblock New York, NY, USA: Wiley-Interscience, July 2006.

\bibitem{DonohMM_Message}
D.~Donoho, A.~Maleki, and A.~Montanari.
\newblock Message passing algorithms for compressed sensing.
\newblock {\em Proc. Natl. Acad. Sci.}, 106(45):18914--18919, Nov. 2009.

\bibitem{VAMP_NIPS}
A.~Fletcher, P.~Pandit, S.~Rangan, S.~Sarkar, and P.~Schniter.
\newblock Plug-in estimation in high-dimensional linear inverse problems: A
  rigorous analysis.
\newblock In {\em Workshop Neural Info. Proc. Sys. (NIPS)}, pages 7451--7460,
  2018.

\bibitem{Gallager68}
R.~Gallager.
\newblock {\em Information Theory and Reliable Communications}.
\newblock Wiley, Jan. 1968.

\bibitem{GuoVerdu2005}
D.~Guo and S.~Verd{\'u}.
\newblock Randomly spread {CDMA}: {A}symptotics via statistical physics.
\newblock {\em IEEE Trans. Inf. Theory}, 51(6):1983--2010, June 2005.

\bibitem{Herzat2013}
C.~Herzet, C.~Soussen, J.~Idier, and R.~Gribonval.
\newblock Exact recovery conditions for sparse representation with partial
  support information.
\newblock {\em IEEE Trans. Inf. Theory}, 59(11):7509--7524, Aug. 2013.

\bibitem{kamilov2013sparse}
U.~Kamilov, A.~Bourquard, and M.~Unser.
\newblock Sparse image deconvolution with message passing.
\newblock In {\em Proc. 5th Workshop on Signal Process. with Adaptive Sparse
  Structured Representations (SPARS)}, Feb. 2013.

\bibitem{Kamilov2012}
U.~Kamilov, S.~Rangan, A.~Fletcher, and M.~Unser.
\newblock Approximate message passing with consistent parameter estimation and
  applications to sparse learning.
\newblock In {\em Workshop Neural Info. Proc. Sys. (NIPS)}, pages 2447--2455,
  Dec. 2012.

\bibitem{Krzakala2012probabilistic}
F.~Krzakala, M.~M{\'e}zard, F.~Sausset, Y.~Sun, and L.~Zdeborov{\'a}.
\newblock Probabilistic reconstruction in compressed sensing: {A}lgorithms,
  phase diagrams, and threshold achieving matrices.
\newblock {\em J. Stat. Mech. - Theory E.}, 2012(08):P08009, Aug. 2012.

\bibitem{Lloyd82}
S.~Lloyd.
\newblock Least squares quantization in {PCM}.
\newblock {\em IEEE Trans. Inf. Theory}, 28(2):129--137, Mar. 1982.

\bibitem{VanLuong2017}
H.~V. Luong, J.~Seiler, A.~Kaup, S.~Forchhammer, and N.~Deligiannis.
\newblock Measurement bounds for sparse signal reconstruction with multiple
  side information.
\newblock {\em Arxiv preprint arXiv:1605.03234}, Jan. 2017.

\bibitem{MacCartney_2017}
G.~MacCartney, T.~Rappaport, and S.~Rangan.
\newblock Rapid fading due to human blockage in pedestrian crowds at 5g
  millimeter-wave frequencies.
\newblock {\em GLOBECOM 2017 - 2017 IEEE Global Communications Conference}, Dec
  2017.

\bibitem{maleki2010approximate}
A.~Maleki.
\newblock {\em Approximate message passing algorithms for compressed sensing}.
\newblock Stanford University, Nov. 2010.

\bibitem{manoel2017streaming}
A.~Manoel, F.~Krzakala, E.~W. Tramel, and L.~Zdeborov{\'a}.
\newblock Streaming bayesian inference: theoretical limits and mini-batch
  approximate message-passing.
\newblock In {\em Communication, Control, and Computing (Allerton), 2017 55th
  Annual Allerton Conference on}, pages 1048--1055. IEEE, 2017.

\bibitem{SaabM_weighted}
H.~Mansour and R.~Saab.
\newblock Recovery analysis for weighted $\ell_1$-minimization using the null
  space property.
\newblock {\em Appl. Comput. Harmon. Anal.}, 43(1):23--38, July 2017.

\bibitem{Mota2017}
J.~Mota, N.~Deligiannis, and M.~Rodrigues.
\newblock Compressed sensing with prior information: {S}trategies, geometry,
  and bounds.
\newblock {\em IEEE Trans. Inf. Theory}, 63(7):4472--4496, July 2017.

\bibitem{Mota2016}
J.~Mota, N.~Deligiannis, A.~Sankaranaraynan, V.~Cevher, and M.~Rodrigues.
\newblock Adaptive-rate reconstruction of time-varying signals with application
  in compressive foreground extraction.
\newblock {\em IEEE Trans. Signal Process.}, 64(14):3651--3666, Mar. 2016.

\bibitem{NeedellWeighted16}
D.~Needell, R.~Saab, and T.~Woolf.
\newblock Weighted-minimization for sparse recovery under arbitrary prior
  information.
\newblock {\em Inst. Math. Inf. Infer.}, 6(3):284--309, Jan. 2017.

\bibitem{RanganGAMP2010}
S.~Rangan.
\newblock Generalized approximate message passing for estimation with random
  linear mixing.
\newblock In {\em Proc. IEEE Int. Symp. Inf. Theory (ISIT)}, pages 2168--2172,
  July 2011.

\bibitem{RanganADMMGAMP2015_ISIT}
S.~Rangan, A.~Fletcher, P.~Schniter, and U.~Kamilov.
\newblock Inference for generalized linear models via alternating directions
  and {B}ethe free energy minimization.
\newblock In {\em Proc. Int. Symp. Inf. Theory (ISIT)}, pages 1640--1644, June
  2015.

\bibitem{Rangan2014}
S.~Rangan, P.~Schniter, and A.~Fletcher.
\newblock On the convergence of approximate message passing with arbitrary
  matrices.
\newblock In {\em Proc. IEEE Int. Symp. Inform. Theory (ISIT)}, pages 236--240,
  Feb. 2014.

\bibitem{RanganVAMP}
S.~Rangan, P.~Schniter, and A.~Fletcher.
\newblock Vector approximate message passing.
\newblock In {\em Proc. IEEE Int. Symp. Inf. Theory (ISIT)}, pages 1588--1592,
  July 2017.

\bibitem{reeves2016replica}
G.~Reeves and H.~D. Pfister.
\newblock The replica-symmetric prediction for compressed sensing with gaussian
  matrices is exact.
\newblock In {\em Proc. IEEE Int. Symp. Inform. Theory (ISIT)}, pages 665--669.
  IEEE, 2016.

\bibitem{Renna2016}
F.~Renna, L.~Wang, X.~Yuan, J.~Yang, G.~Reeves, A.~Calderbank, L.Carin, and
  M.~Rodrigues.
\newblock Classification and reconstruction of high-dimensional signals from
  low-dimensional features in the presence of side information.
\newblock {\em IEEE Trans. Inf. Theory}, 62(11):6459--6492, Sept. 2016.

\bibitem{RushV16}
C.~Rush and R.~Venkataramanan.
\newblock Finite sample analysis of approximate message passing.
\newblock {\em IEEE Trans. Inf. Theory}, (forthcoming).
\newblock Available: \url{https://ieeexplore.ieee.org/document/8318695/}.

\bibitem{saleh1987statistical}
A.~Saleh and R.~Valenzuela.
\newblock A statistical model for indoor multipath propagation.
\newblock {\em IEEE J. Select. Areas Commun.}, 5(2):128--137, Feb. 1987.

\bibitem{FadeShadow}
P.~Shankar.
\newblock {\em Fading and Shadowing in Wireless Systems}.
\newblock Springer, 2 edition, 2019.

\bibitem{takhar2006new}
D.~Takhar, J.~Laska, M.~Wakin, M.~Duarte, D.~Baron, S.~Sarvotham, K.~Kelly, and
  R.~Baraniuk.
\newblock A new compressive imaging camera architecture using optical-domain
  compression.
\newblock Feb. 2006.

\bibitem{Tanaka2002}
T.~Tanaka.
\newblock {A statistical-mechanics approach to large-system analysis of CDMA
  multiuser detectors}.
\newblock {\em IEEE Trans. Inf. Theory}, 48(11):2888--2910, Nov. 2002.

\bibitem{LASSO1996}
R.~Tibshirani.
\newblock Regression shrinkage and selection via the {LASSO}.
\newblock {\em J. Royal Stat. Soc. Series B (Methodological)}, 58(1):267--288,
  Jan. 1996.

\bibitem{Vaswani2010}
N.~Vaswani and W.~Lu.
\newblock Modified-{CS}: Modifying compressive sensing problems for partially
  known support.
\newblock {\em IEEE Trans. Signal Process.}, 58(9):4595--4607, May 2010.

\bibitem{wang2015approximate}
X.~Wang and J.~Liang.
\newblock Approximate message passing-based compressed sensing reconstruction
  with generalized elastic net prior.
\newblock {\em Signal Process. Image}, 37:19--33, Sept. 2015.

\bibitem{Weizman2015}
L.~Weizman, Y.~Eldar, and D.~Bashat.
\newblock Compressed sensing for longitudinal {MRI}: An adaptive-weighted
  approach.
\newblock {\em Medical Physics}, 42(9):5195--5208, Nov. 2015.

\bibitem{ZhuBaronMPAMP2016ArXiv}
J.~Zhu, D.~Baron, and A.~Beirami.
\newblock Optimal trade-offs in multi-processor approximate message passing.
\newblock {\em Arxiv preprint arXiv:1601.03790}, Nov. 2016.

\bibitem{DCSAMP}
J.~Ziniel and P.~Schniter.
\newblock Dynamic compressive sensing of time-varying signals via approximate
  message passing.
\newblock {\em IEEE Trans. Signal Process.}, 61(21):5270--5284, Nov. 2013.

\end{thebibliography}

\end{document}